\documentclass{article}

\usepackage[final]{corl_2020} 
\usepackage{wrapfig}

\usepackage{graphicx}
\usepackage{amsmath}
\usepackage{amsthm}
\usepackage{bm}

\usepackage{tabularx}
\newcolumntype{C}{>{\centering\arraybackslash}p{19mm}}
\newcolumntype{G}{>{\centering\arraybackslash}p{4mm}}
\newcolumntype{S}{>{\centering\arraybackslash\scriptsize}p{4mm}}

\usepackage[font=small,labelfont=bf]{caption}
\usepackage{color}

\usepackage{amsthm}

\usepackage{cleveref} 

\usepackage{xspace}

\usepackage{enumitem} 

\usepackage{cite} 

\usepackage{amssymb}
\usepackage{algorithm}
\makeatletter
\renewcommand*{\ALG@name}{Alg.}
\makeatother
\usepackage[noend]{algpseudocode}

\newtheoremstyle{break}
  {}
  {}
  {\itshape}
  {}
  {\bfseries}
  {.}
  {\newline}
  {}

\theoremstyle{break}
\newtheorem{theorem}{Theorem}

\newcommand\mydots{\hbox to 1em{.\hss.\hss.}}

\newcommand{\x}{\bm{x}} 
\newcommand{\bc}{\bm{c}} 
\newcommand{\ac}{\bm{u}} 
\newcommand{\z}{\bm{z}} 
\newcommand{\bmu}{\boldsymbol\mu} 

\newcommand{\bQ}{\bm{Q}}

\newcommand{\w}{\bm{w}} 
\newcommand{\btheta}{\bm{\theta}} 
\newcommand{\lr}{\bm{\eta}} 
 
\newcommand{\cost}{l}

\newcommand{\f}{\bm{f}} 
\newcommand{\h}{\bm{h}} 
\newcommand{\g}{\bm{g}} 

\newcommand{\Prob}{\mathbb{P}}

\newcommand{\bq}{\mathbf{q}}
\newcommand{\bv}{\mathbf{v}}
\newcommand{\bomega}{\boldsymbol\omega}
\newcommand{\bF}{\mathbf{F}}
\newcommand{\bM}{\mathbf{M}} 
\newcommand{\bJ}{\mathbf{J}}

\newcommand{\bS}{\mathbf{S}}
\newcommand{\ba}{\mathbf{a}} 
\newcommand{\bpos}{\mathbf{p}} 
\newcommand{\pos}{\bm{p}}
\newcommand{\vel}{\bm{v}}

\newcommand{\B}{\mathcal{B}}
\newcommand{\C}{\mathcal{C}}
\newcommand{\X}{\mathcal{X}}

\newcommand{\Xsafe}{\mathcal{X}_{\text{free}}}
\newcommand{\Xgoal}{\mathcal{X}_{\text{goal}}}
\newcommand{\U}{\mathcal{U}}

\newcommand{\Z}{\mathcal{Z}}
\newcommand{\Loss}{\mathcal{L}}

\newcommand{\R}{\mathbb{R}}
\newcommand{\W}{\mathbb{W}}

\newtheoremstyle{nolinebreakstyle}
  {}
  {}
  {} 
  {} 
  {\bfseries} 
  {.} 
  {.2em} 
  {} 
\theoremstyle{nolinebreakstyle}
\newtheorem{mydef}{Definition}

\newtheoremstyle{exampstyle}
  {1em plus .2em minus .1em}
  {1em plus .2em minus .1em}
  {} 
  {} 
  {\bfseries} 
  {.} 
  {.5em} 
  {} 
  
\theoremstyle{exampstyle}

\theoremstyle{exampstyle}

\theoremstyle{exampstyle}
\newtheorem{preremark3}{Theorem}

\theoremstyle{exampstyle}

\theoremstyle{exampstyle}

\theoremstyle{exampstyle}

\newcommand\Tstrut{\rule{0pt}{2.6ex}}         
\newcommand\Bstrut{\rule[-0.9ex]{0pt}{0pt}}   
\usepackage{multirow}

\title{Sampling-based Reachability Analysis: A Random Set Theory Approach with Adversarial Sampling}

\author{%
    Thomas Lew, Marco Pavone\\
    Department of Aeronautics and Astronautics\\
    Stanford University,  
    United States\\
    \texttt{\{thomas.lew,pavone\}@stanford.edu}
}

\begin{document}
\maketitle

\begin{abstract}    
Reachability analysis is at the core of many applications, from neural network verification, to safe trajectory planning of uncertain systems. However, this problem is notoriously challenging, and current approaches tend to be either too restrictive, too slow, too conservative, or approximate and therefore lack guarantees. In this paper, we propose a simple yet effective sampling-based approach to perform reachability analysis for arbitrary dynamical systems. Our key novel idea consists of using random set theory to give a rigorous interpretation of our method, and prove that it returns sets which are guaranteed to converge to the convex hull of the true reachable sets. Additionally, we leverage recent work on robust deep learning and propose a new adversarial sampling approach to robustify our algorithm and accelerate its convergence. We demonstrate that our method is faster and less conservative than prior work, present results for approximate reachability analysis of neural networks and robust trajectory optimization of high-dimensional uncertain nonlinear systems, and discuss future applications\footnote{All code is available at \href{https://github.com/StanfordASL/UP}{https://github.com/StanfordASL/UP}}.

\end{abstract}

\keywords{Reachability analysis, robust planning and control, neural networks}

\section{Introduction}\label{sec:intro}


Reachability analysis is at the core of many applications, from robust trajectory planning of uncertain systems, to neural network verification. 
Generally, 
it entails characterizing the set of reachable states for a system at any given time in the future. For instance, planning a trajectory for a quadrotor carrying a payload of uncertain mass in severe wind requires ensuring that no reachable state collides with obstacles. 
%
In formal verification of neural networks, reachability analysis can be used to quantify the change in output for various input perturbations, and hence ensure prediction accuracy despite adversarial examples. 
However, reachability analysis is notoriously challenging, as it requires describing all reachable states from \emph{any} possible initial state and any realization of uncertain parameters of the system. In contrast to approaches which can handle problems with known probability distributions over parameters, e.g., when the state of the system is estimated with Kalman filtering, and parameters are updated through Bayesian inference, sometimes only bounds on unknown parameters are available, which is the case when constructing confidence sets for the parameters of the model  \citep{abbasi2011improved}. 
To perform reachability analysis for problems with bounded uncertainty, current methods tend to be encumbered by strong assumptions which are difficult to verify in practice, do not scale well to complex systems, or are too slow to be used within data-driven controllers which use and refine bounds on model parameters in real time. In practice for robotic applications, they often require tuning parameters used to provide theoretical guarantees and yield optimal performance, or using a simplified model of the system, e.g., assuming disturbances affect the system additively. 
These are reasonable assumptions for many applications, but the general problem remains a challenge.

\begin{figure}[!htb]
    \centering
    \includegraphics[width=1\linewidth]{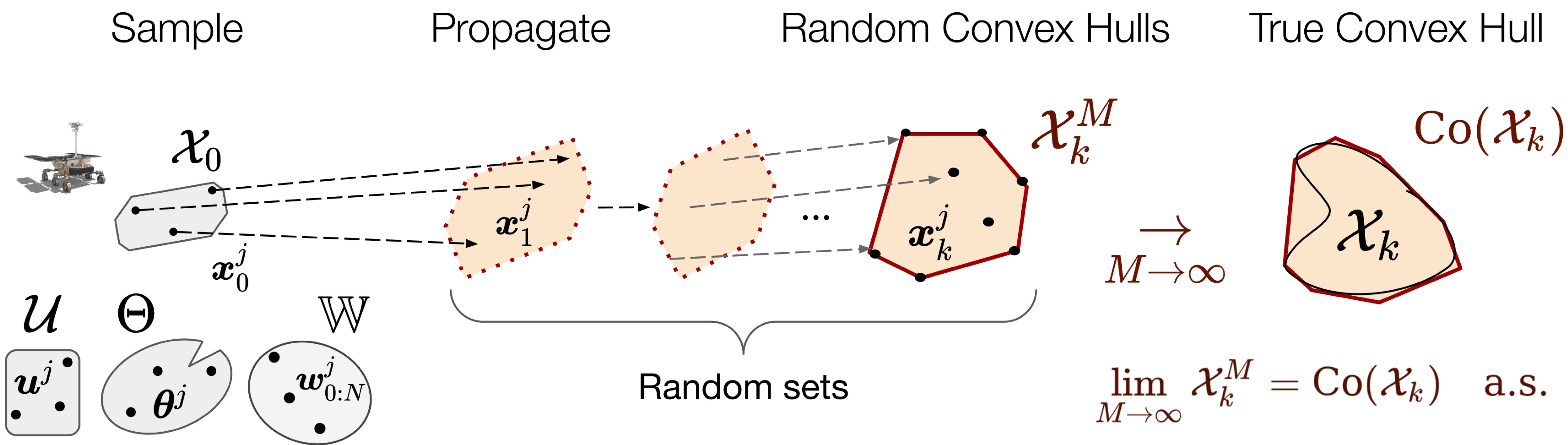}
    \caption{(\textbf{randUP}) consists of three steps: 1) initial states, controls, disturbances and parameters are sampled, 2) each particle is propagated according to the nonlinear dynamics, and 3) the convex hull at each time $k$ is computed. Since each convex hull depends on randomly sampled parameters, it is itself random, and can be mathematically described using theory of random sets \citep{Molchanov_BookTheoryOfRandomSets2017}. In Theorem \ref{thm:conv_convHulls_randSets}, we prove that each convex hull $\X_k^M(\omega)$ is guaranteed to converge to the true convex hull of the reachable set $\X_k$ as $M{\rightarrow}\infty$. In Section \ref{sec:UP_adv}, we propose (\textbf{robUP!}): an adversarial sampling scheme robustifying (\textbf{randUP}) by accelerating its convergence.}
    \label{fig:random_set_prop}
    \vspace{-5mm}
\end{figure}

\vspace{-2mm}\paragraph{Contributions} In this work, we consider the problem of reachability analysis for general systems with bounded state and parameters uncertainty, and bounded external disturbances. Our proposed approach can tackle high-dimensional nonlinear dynamics, and makes minimal assumptions about the properties of the system. It is fast, parallelizable, simple to implement, and provides accurate estimates of reachable sets. Our key novel technical idea consists of leveraging random set theory to provide theoretical convergence guarantees of our algorithm, and paves the way for future applications of this theory to the field of robotics. Specifically, our contributions are as follows: 
\begin{enumerate}[leftmargin=6mm]
    \item A novel, simple, yet effective sampling-based method to perform reachability analysis for general nonlinear systems and bounded uncertainty. Using the theory of random sets, we prove that our method is guaranteed to asymptotically converge to the convex hull of the reachable sets.

    \item We develop a novel adversarial sampling scheme inspired by the literature on robust deep learning to accelerate convergence of our algorithm, and identify its strengths and weaknesses.

    \item We demonstrate our method on a neural network, and on the robust planning problem of a 13-dimensional nonlinear uncertain spacecraft. We show that our method outperforms current approaches by yielding tighter reachable sets with faster computation time. 

    \item We discuss immediate applications of our method, as well as possible future applications of random set theory for robotics and deep learning.
\end{enumerate}


\vspace{-2mm}\paragraph{Related work} When probability distributions over uncertain parameters are available, sampling-based methods can propagate uncertainty \citep{SchmerlingPavone2017,LoquercioRAL2020} with asymptotic guarantees. Scenario optimization can also be used for reachability analysis \citep{EstReachSets_L4DC20}, and ensure probabilistic convex constraints satisfaction \citep{EsfahaniScenario2015}. Although various methods exist to approximately tackle such problems \citep{hewing2018cautious,LewBonalli2020,CastilloRAL2020,BerretAutomatica2020}, accurate probability distributions over parameters are required to ensure safe operation of the underlying autonomous systems. Instead, the problem of reachability analysis with bounded uncertainty is of interest whenever no accurate prior is available, and one must rely on confidence sets for the  parameters of the model \citep{abbasi2011improved}. For such problems, endowing a sampling-based approach for reachability analysis with guarantees is a challenging task; this often requires precise knowledge of certain mathematical properties of the system, e.g., the Lipschitz constant of the reachability function from any state with respect to the Hausdorff distance, the ratio of the surface area of the true reachable set to its volume \citep{LiebenweinRSS2018}, or smoothness properties of the boundary of the true reachable set \citep{Kor1993_BookMinimaxImgReconstr}. 

Alternatively, methods which directly leverage smoothness properties of dynamical systems are capable of computing outer-approximations of reachable sets. 
Such algorithms often leverage precise knowledge of the Lipschitz constant of the system to bound the Lagrange remainder resulting from a finite order Taylor series approximation \citep{koller2018,Dean2019RobustGF}. 
Unfortunately, precise Lipschitz constants are rarely available in data-driven control applications, and even with this knowledge these algorithms are often too conservative in practice (see Section \ref{sec:result}).
Computing tight upper-bounds on Lipschitz constants of neural networks is an active field of research \citep{fazlyab2019efficient,lipsPoly2020}, but current methods are computationally expensive, 
which prevents their use in learning-based control where the model of the system is updated online (e.g., via gradient descent \citep{MAML2017}, or linear regression over the last layer \citep{HarrisonSharmaEtAl2018,LewSharmaHarrisonPavone2020}). 
Recently, efficient and scalable sampling-based algorithms have been proposed to estimate Lipschitz constants  \citep{Dean2019RobustGF,WengExtreme2018}, but such  methods are not guaranteed to provide upper bounds \citep{fazlyab2019efficient}.

Finally, Hamilton-Jacobi (HJ) reachability analysis can compute reachable sets exactly \citep{HJIoverviewBansal2017,FisacRSS2018}. Unfortunately, handling arbitrary systems 
(e.g., neural networks) 
with such approaches is challenging, as they require  
solving a partial differential equation involving a max/minimization over controls and disturbances. For this reason, these methods typically discretize the state space and thus suffer from the curse of dimensionality \citep{HJIoverviewBansal2017}. Further, these methods leverage the principle of dynamic programming to compute solutions, and thus 
cannot handle parameter uncertainty which introduces time correlations along the trajectory (see Section \ref{sec:problem}). 
%
Other approaches include using surrogate models to directly parameterize reachable sets \citep{RubiesICRA2019,fan2020deep},  
and formal verification tools for dynamical systems \citep{ChenFlowstar2012} and neural networks \citep{IvanovVerisig2019}. 
As these methods typically perform all computation offline\footnote{%
Although recent work enables warm-starting HJ reachability analysis \citep{warmStartingHJI_Herbert2019}, such methods still require a few seconds for a 3D linear system, and more than an hour for a 10D quadrotor.}
, they may not be adequate for applications 
where bounds on model parameters are updated over time.

Unlike prior work, our method treats the general problem of performing real-time multi-steps reachability analysis for arbitrary systems. 
Specifically, we only assume that the dynamics are continuously differentiable, and that the uncertainty is bounded. 
To do so, we leverage random set theory to provide asymptotic convergence guarantees to the convex hull of the reachable sets. 
Our proposed method is suitable to various applications where approximations are sufficient in practice, and stronger guarantees can be derived given further system-specific assumptions (see Section \ref{sec:applications}). 

{
\small
\vspace{-2mm}\paragraph{Notation} We denote the sets of natural and real numbers by $\mathbb{N}$ and $\R$, respectively, and the Borel $\sigma$-algebra in $\R^n$ by $\B(\R^n)$. We use $\mathcal{F}$, $\mathcal{G}$, and $\mathcal{K}$ to denote the families of closed, open, and compact sets in $\R^n$, respectively. We write $\mathcal{S}^N{:=}\,\mathcal{S}{\times}{\mydots}{\times}\mathcal{S}$ ($N$ times) for any set $\mathcal{S}$, and $\smash{\|\x\|^2_{\bQ}}\,{=}\,\x^T\bQ\x$ with positive-definite matrix $\bQ$. The function composition operator is denoted as $\circ$, and the convex hull of a discrete set $\smash{\mathcal{X}^M}{=}\smash{\{\x^j\}_{j=1}^M}$ as $\smash{\textrm{Co}(\mathcal{X}^M)}$.
}%

\section{Problem Formulation}\label{sec:problem}
We consider general discrete-time dynamical systems of the form
\begin{equation}\label{eq:dyns}
    \x_{k+1} = \f(\x_k, \ac_k, \btheta, \w_k),
\end{equation}
with state $\x_k \in \R^n$, control input $\ac_k \in \U_k$, uncertain parameters $\btheta \in \Theta$, external disturbance $\w_k \in \W$, and initial state $\x_0 \in \X_0$, where $\U_k \subset \R^m$, $\Theta \subset \R^p$, ${\W \subset \R^q}$, and $\X_0 \subset \R^n$ are known compact sets. 
We only assume of the dynamics $\f(\cdot)$ to be continuously differentiable, which includes neural network approximators and common systems in robotic applications.


Let $N\in\mathbb{N}$, a finite time horizon. 
The goal of this paper consists of performing reachability analysis for \eqref{eq:dyns},  by computing reachable sets $\X_1,\dots,\X_N$ in which the state trajectory is guaranteed to lie.  
Formally, the reachable set $\X_k$ at each time $k=1,\dots,N$ can be expressed as
%
%
\begin{align}
\hspace{-2mm}
\X_k = \bigg\{ 
	\x_k = \f(\cdot, \ac_{k-1}, \btheta, \w_{k-1})\circ\dots\circ\f(\x_0, \ac_0, \btheta, \w_0)
\ \bigg| 
\begin{array}{l}
\x_0\in\X_0, \,
\ac_i\in\U_i, \, \w_i\in\W \\
\btheta\in\Theta,
\quad \ \ \mbox{\footnotesize $i=1,\mydots,k-1$} 
\end{array}
\bigg\}
\label{eq:reach_set_onestep}
.
\end{align}
%
%
Specifically, $\X_k$ describes the set of all possible reachable states at time $k$, using control inputs respecting  actuator constraints, and for any possible model parameter and external disturbance. 
This problem formulation also allows the verification of a given feedback controller   $\ac\,{=}\,\boldsymbol\kappa(\x)$, by computing the reachable set of the system $\f(\x_k, \boldsymbol\kappa(\x_k), \btheta, \w_k)$. 
Also, evaluating the reachable set for a given sequence of open-loop control inputs 
$(\ac_0,\mydots,\ac_{N{-}1})$ 
is a specific instance of this problem, with $\U_i\,{=}\,\{\ac_i\}$. 
By continuity of $\f$ and compactness of $\X_0$, $\U_i$, $\Theta$ and $\W$, all reachable sets 
are guaranteed to be compact. 
%
We note that \eqref{eq:reach_set_onestep} is different than defining the following recursion \citep{rosolia2019}
\begin{equation*}
\smash{
\tilde\X_{k+1}=\{ 
\x_{k+1} \ | \ \x_{k+1} = \f(\x_k, \ac_k, \btheta, \w_k), \ 
\x_k\in\tilde\X_k, \
\ac_k\in\U_k, \ 
\btheta\in\Theta, \ 
\w_k\in\W
\}
, 
\
\tilde\X_0 =\X_0,
}
\end{equation*}
which neglects the time dependency of the trajectory on $\btheta$. For instance, consider $\x_{k+1}=\btheta\x_k$, with $\btheta\in\{1,2\}$, and $\X_0=\{1\}$. Then, $\X_2=\{1,4\}$. However,  $\smash{\tilde\X_2}=\{1,2,4\}$, which is artificially more conservative than using \eqref{eq:reach_set_onestep}. This is an issue which many uncertainty propagation methods suffer from \citep{koller2018,fan2020deep,LewBonalli2020}, as it can cause additional conservatism or innacuracies \citep{simGPs_L4DC20} when considering uncertainty over model parameters.

    
Before proposing an algorithm designed to perform reachability analysis, we recall the three main difficulties of this problem: 
a) minimal assumptions are made about the dynamics \eqref{eq:dyns}; 
b) 
reachable sets should be computed as fast as possible, to enable data-driven applications where model parameters are updated online, and to embed this reachability tool within control feedback loops for robotic applications; 
c) The method should scale to relatively high dimensions $n,m,p,q$. In the following section, we present a sampling-based methodology which addresses these challenges.

\section{Approximate Reachability Analysis using Random Set Theory}\label{sec:UP_randSets}

\vspace{-1mm}

Consider computing $\X_k$ exactly. To design an algorithm, we start with the observation that if we could evaluate all possible values of $(\x_0,\ac,\btheta,\w)\in$ 
$\X_0{\times}\U^{\scalebox{.6}{$k{-}1$}}{\times}\Theta{\times}\W^{\scalebox{.6}{$k{-}1$}}$, 
and compute $\x_k$ for each tuple by forward propagation through the dynamics, then $\X_k$ would be known exactly. Unfortunately, this is only possible if $\X_0,\U,\Theta$, and $\W$ are finite and small, which is generally not the case.

\vspace{-1mm}

Instead, consider sampling a finite number~$M$ of initial states~$\x_0^j$, control trajectories~$\ac^j$, parameters~$\btheta^j$, and disturbances~$\w^j$, resulting in a finite number of states~$\smash{\x_k^j}$ which are guaranteed to lie within the true unknown reachable set~$\X_k$. From this finite set of states, three questions arise:
\vspace{-2mm}
\begin{itemize}
    \item How can we best approximate $\X_k$ from $\smash{\{\x_k^j\}_{j=1}^M}$?

    \item What theoretical guarantees result from such a sampling-based approach?
    
    \item How can we best select samples $\smash{\{\x_0^j,\ac^j,\btheta^j,\w^j\}_{j=1}^M}$, for both efficiency and accuracy?
\end{itemize}
\vspace{-1.5mm}
This section treats the first two questions, and the third will be dealt with in Section \ref{sec:UP_adv}. Before proposing a solution, it is important to correctly characterize the properties of possible set estimators which leverage the samples $\smash{\{\x_k^j\}_{j=1}^M}$. Although different approaches to approximate sets from samples can be found in the literature, they typically assume that samples are drawn uniformly within the set itself, and leverage smoothness properties of the set boundaries to provide convergence guarantees of their estimators \citep{Kor1993_BookMinimaxImgReconstr,RipleyPoissonForest1977,Rodriguez2016}. To treat our more general problem formulation, we opt for an alternative mathematical description of the class of estimators computed from a collection of randomly sampled states $\smash{\x^j}$: they are \textit{random sets}  \citep{Molchanov_BookTheoryOfRandomSets2017}. Intuitively, different approximate reachable sets will be computed for different realizations of the samples $\smash{(\x_0^j,\ac^j,\btheta^j,\w^j)}$. Mathematically, two random sets share the same distribution if their probabilities to intersect any compact set $K\,{\subset}\,\R^n$ are equal. A random set is formally defined as a map from a probability space to a family of sets \citep{Molchanov_BookTheoryOfRandomSets2017}:

\vspace{0.6mm}

\begin{mydef}[Random Set]\label{def:rand_set}
    A map $\X : \Omega \rightarrow \mathcal{F}$ from a probability space $(\Omega,\mathcal{A},\Prob)$ to the family of closed sets in $\R^n$ is a random closed set if $\{\omega \mid \mathcal{X}(\omega) \cap \mathcal{K} \neq \emptyset\} \in \mathcal{A}$ for every compact set $K \subset \R^n$.
\end{mydef}
\vspace{-0.7mm}

In this work, we seek convex approximations of the reachable sets, as convex sets are widely used for control applications and neural network verification. Specifically, we consider using the convex hull of the samples $\smash{\{\x_k^j\}_{j=1}^M}$. The interpretation of the convex hull of sampled points as a random set was mentioned in \citep{Molchanov_BookTheoryOfRandomSets2017}, but to the best of our knowledge, no prior work uses the theory of random sets for reachability analysis. We leverage random set theory to enable a clear analysis of the properties of our method, despite very weak assumptions about our system. 

\begin{minipage}{\linewidth}
	\begin{wrapfigure}{R}{0.54\linewidth}
		\hspace{2mm}
		\begin{minipage}{0.95\linewidth}
	\vspace{-25pt}
    \begin{algorithm}[H]
    	\small
		\caption{UP using Random Sampling (\textbf{randUP})} \label{alg:randUP}
	    \vspace{1mm}
	    \textbf{Parameters}: 
	    Number of samples $M$
	    \\[1mm]
		\textbf{Output}: 
	    Convex approximation of reachable sets $\X_k$
		\\[-1mm]
		\hrule\vspace{1mm}
	\begin{algorithmic}[1]
		\State Sample i.i.d. $\smash{(\x_0^j,\ac^j,\btheta^j,\w^j)}$, 
		\qquad\ \ \ 
		$j=1,\mydots,M$
		\State $\x^j_{1:N}\gets  \textrm{Propagate}(\x_0^j,\ac^j,\btheta^j,\w^j)$,
		$\ \, j=1,\mydots,M$
		\State $\X_k^M = \textrm{Co}(\{\x_k^j\}_{j=1}^M)$, 
		$\qquad\qquad\quad\ \ \ \, \, k=1,\mydots,N$
		\State\Return $\{\X_k^M\}_{k=1}^N$
	\end{algorithmic}
	\end{algorithm}
	\end{minipage}%
	\end{wrapfigure}
    We propose the simple sampling-based procedure in Algorithm~\ref{alg:randUP} illustrated in~Figure~\ref{fig:random_set_prop} to approximate the reachable sets. 
    It consists of 
     1) sampling i.i.d.~tuples $\smash{(\x_0^j,\ac^j,\btheta^j,\w^j){\in}\X_0{\times}\U^{\scalebox{.6}{$N$}}{\times}\Theta{\times}\W^{\scalebox{.6}{$N{-}1$}}}$ according to arbitrary probability distributions over $\X_0$, $\U$, $\Theta$ and $\W$, 2) propagating each sample through the dynamics~\eqref{eq:dyns} to obtain states within the true reachable sets $\X_k$, and 3) taking their convex hull. 
    Outer ellipsoidal sets or zonotopes can also be used as alternatives to facilitate the downstream application (see Section~\ref{sec:applications}). 
    Different domain-specific sampling distributions are possible; e.g., a beta distribution of shape parameters $\alpha\,{=}\,\beta\,{\ll}\,1$ for additive disturbances may maximize the volume of $\smash{\X_k^M}$. 
\end{minipage}


The key advantages of (\textbf{randUP}) over standard reachability analysis tools are that 
1) it is easy to understand and implement, 
2) the number of particles is a design choice, which can be chosen to best exploit available computation resources, or meet a precision criteria,
3) it is parallelizable, and 
4) it requires minimal assumptions about the system, namely continuity of the dynamics, and boundedness of the parameter space. 
To clarify, $\f\in\C^1$ in \eqref{eq:dyns} is only required for the algorithm introduced in the next section, which will leverage gradient information to further improve robustness.

Despite the simplicity of our algorithm, it is guaranteed to converge to the convex hulls of the true reachable sets $\X_k$. 
From a theoretical point of view, this result is important, as it guarantees that by increasing the number $M$ of particles, the true reachable sets $\X_k$ will be contained within the sampled convex hulls. 
At the core of our proof of this result lie mathematical tools from the theory of random sets. 
Specifically, we first 
restate a theorem in \citep[Prop. 1.7.23]{Molchanov_BookTheoryOfRandomSets2017}, providing necessary and sufficient conditions of convergence of random sets which are key to derive our main result:


\begin{minipage}{\linewidth}
	\begin{wrapfigure}{R}{0.28\linewidth}
		\hspace{4mm}
		\begin{minipage}{1\linewidth}
	\vspace{-2.25cm}  
	\centering
	\includegraphics[width=0.56\linewidth]{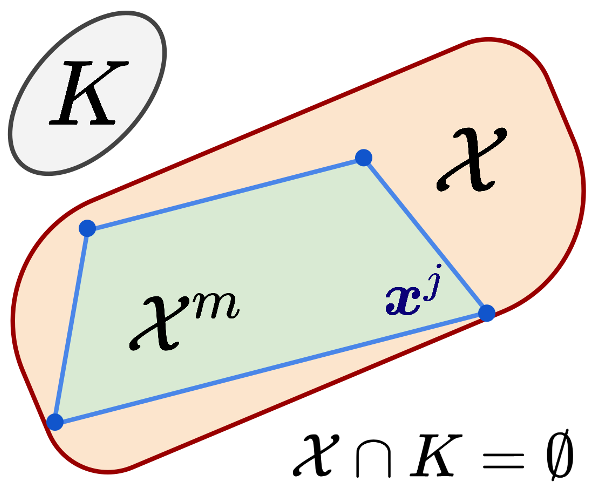}
	
	\centering
	\includegraphics[width=0.56\linewidth]{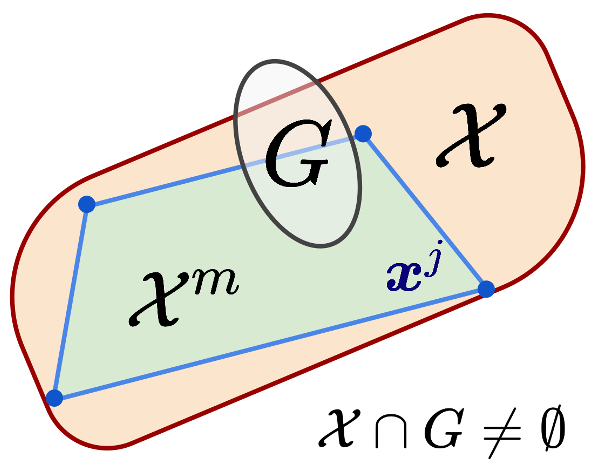}
	\vspace{-2mm}
	\caption{Conditions (\textbf{C1-2}).}
	\end{minipage}%
	\end{wrapfigure}
	\begin{minipage}{1\linewidth}
    \begin{theorem}[Convergence of Random Sets to a Deterministic Limit\vspace{2mm}]\label{thm:conv_randSets_detLim}

A sequence $\{\X^m\}_{m=1}^\infty$ of random closed sets in $\R^n$ 
almost surely converges to a deterministic closed set $\X$ if and only if the following conditions hold:

(\textbf{C1}) For any $K\in\mathcal{K}$, with $\mathcal{K}$ the family of compact sets in $\R^n$,
\begin{equation}
\X\cap K = \emptyset \implies 
\Prob(\X^m\cap K \neq \emptyset \ \ \textrm{infinitely often}) =
0.
\end{equation}

(\textbf{C2}) For any $G\in\mathcal{G}$, with $\mathcal{G}$ the family of open sets in $\R^n$,
\begin{equation}
\X\cap G \neq \emptyset \implies 
\Prob(\X^m\cap G = \emptyset \ \ \textrm{infinitely often}) = 
0.
\end{equation}

\end{theorem}
	\end{minipage}%
\end{minipage}

%
%
%



With this theorem, we prove that our method in Algorithm \ref{alg:randUP} computing approximations of the reachable sets is guaranteed to converge to the convex hull of the true reachable sets.

\begin{theorem}[Convergence of Random Convex Hulls using (\textbf{randUP})\vspace{1mm}] 
\label{thm:conv_convHulls_randSets}
	Let 
	$\smash{\{(\x_0^j,\ac^j,\btheta^j,\w^j)\}_{j=1}^m}$
	 be i.i.d. sampled parameters in $\X_0{\times}\U^{\scalebox{.6}{$k{-}1$}}{\times}\Theta{\times}\W^{\scalebox{.6}{$k{-}1$}}$.  
	Define $\smash{\x_k^j}$  
	according to 
	\eqref{eq:dyns}, 
	and denote the resulting convex hulls as $\smash{\X_k^m \,{=}\, \textrm{Co}(\{\x_k^j\}_{j=1}^m)}$.
	Assume that the sampling distribution of the parameters satisfies $\smash{\Prob(\x_k^j\,{\in}\, G_k)\,{>}\,0}$ for any open set $G_k$ s.t. $\X_k\,{\cap}\,G_k\,{\neq}\, \emptyset$. 
	\vspace{-1mm}
	
	Then, as $m\rightarrow\infty$, $\X_k^m$ converges to the convex hull of the reachable set $\textrm{Co}(\X_k)$ almost surely. 
\end{theorem}
%
%

\vspace{-1mm}
We provide a complete proof in the Appendix and outline it here. First, it consists of showing that $\{\X^m\}_{m=1}^\infty$ is a sequence of compact random sets. Then, we leverage Theorem \ref{thm:conv_randSets_detLim} and show that its conditions are fulfilled. The i.i.d.~assumption is key to apply the second Borel-Cantelli lemma, and the assumption $\smash{\Prob(\x_k^j\,{\in}\, G_k)\,{>}\,0}$ ensures that any point of the true reachable set has a positive probability of being within $\X^m$. By choosing an appropriate sampling distribution over parameters, both assumptions are simple to satisfy in practice. The key advantages of (\textbf{randUP}) are its asymptotic convergence guarantees for general continuous dynamical systems, and that it is intuitive, practical, and can be applied immediately for various applications. Nevertheless, its generality comes with limitations which we discuss below, and relate to existing results in the literature.

\vspace{-2mm}\paragraph{Outer-approximation} For non-convex sets $\X_k$, the convex hull \textrm{Co}$(\X_k)$ is an over-approximation. 
Instead, 
one could take the union of $\epsilon$-balls whose radius $\epsilon$ shrinks depending on the number of samples and the smoothness of the boundary  \citep{Kor1993_BookMinimaxImgReconstr,Rodriguez2016} or dynamics \citep{LiebenweinRSS2018}, which comes with the limitations mentioned in Section \ref{sec:intro}. 
Taking the convex hull is not an issue (and is sufficient) for many applications, e.g., 
uncertainty-aware trajectory optimizers typically use convex hulls to compute feasible obstacle-free paths, as they convexify constraints using first \citep{LewBonalli2020}, or second order information\footnote{In the multimodal case, since  $\f\in\C^1$, multimodal reachable sets  occur if and only if the parameter set is disjoint. In this case, it suffices to run (\textbf{randUP}) on each disjoint set, and return their union.}.

\vspace{-2mm}\paragraph{Inner-approximation} For some classes of reachable sets $\X_k$ (e.g., a ball), our method always returns a subset of $\X_k$. This could be a limitation, especially for safety critical applications. This problem is acknowledged in \citep{RipleyPoissonForest1977}, which proposes to inflate the resulting convex hull by a quantity related to the estimated volume of $\X_k$. Unfortunately, their work focuses on set estimation in two dimensions only, and efficiently computing the volume of an arbitrary convex polytope is challenging \citep{Bueler2000}.
Similar related problems are investigated in \citep[Chap. 5.4-6,7.5]{Kor1993_BookMinimaxImgReconstr}, which proposes maximum likelihood estimators for Dudley sets, as well as star-shaped sets. Unfortunately, their work requires smoothness properties of the boundaries of the estimated set, which are unknown in our problem setting, or would be related to the Lipschitz constant of $\f$ and provide coarse estimations as discussed earlier. Critically, they assume uniform sampling of points directly within $\X_k$. For reachability analysis, parameters 
are sampled, and analytically computing the distribution for $\x_k$ is intractable. 

\vspace{-2mm}\paragraph{Rate of convergence} Although Theorem \ref{thm:conv_convHulls_randSets} provides asymptotic guarantees of convergence to the convex hull of the reachable sets, it does not provide a convergence rate. Rates of convergence for various set estimators are derived in \citep{Kor1993_BookMinimaxImgReconstr,Rodriguez2016}, by leveraging smoothness properties, and uniformly sampling directly within $\X_k$. For our problem setting, uniformly sampling parameters and computing reachable states does not lead to a uniform distribution of samples within $\X_k$, and verifying such smoothness properties is a challenge for reachable sets of general dynamical systems. Nevertheless, we believe that random set theory could be used to analyze convergence rates of similar methods (e.g., using the \textit{variance} of the random sets, as pointed out in \citep{ChevalierKriging2014}), which we leave for future work. 


Next, we propose an adversarial sampling scheme capable of accelerating the convergence of (\textbf{randUP}), alleviating the last two problem-specific limitations above.

\section{Adversarial Sampling for Robust Uncertainty Propagation}\label{sec:UP_adv}

The accuracy of the reachable set approximation depends on the realization of the sampled parameters $\smash{\z^j{=}(\x_0^j,\ac^j,\btheta^j,\w^j)}$. For instance, a sample $\z^{M+1}$ for which $\smash{\x_k^{M+1}\,{\in}\,\X_k^M}$ results in the same convex hull $\smash{\X_k^{M+1}{=}\,\X_k^M}$ and is not informative. This motivates choosing new parameters for which $\smash{\x_k^{M+1}\notin\X_k^M}$. This parallels robust classification and adversarial attacks of neural networks \citep{GoodfellowAdv2015,Kurakin2017,dong2018boosting}. Indeed, the approximate reachable set $\smash{\X_k^M}$ can be interpreted as a classifier of the reachable states, and new samples satisfying  $\smash{\x_k^{M+1}\notin\X_k^M}$ as adversarial examples. The objective then consists of obtaining a robust classifier of the reachable set. 
Adversarial examples for deep neural network classifiers are often found by slightly perturbing inputs to cause drastic changes in the associated predictions. For such parametric models, the training loss explicitly depends on the possibly perturbed input. In contrast, our decision boundary is implicitly defined as a polytope around the sampled points $\smash{\x^j}$, instead of as a parametric function of $\z^j$. Thus, even if a loss function was available for our problem, computing its gradient at any~$\z$ to find adversarial examples is challenging. Moreover, defining a loss function is not trivial; the Hausdorff distance between the approximate and the true reachable sets would be difficult to evaluate since no ground truth of the reachable set is available. In computer vision, the Chamfer distance \citep{Fan2016} or local signed distance functions \citep{Hugues1992,Cohen1998,Chabra2020} are used for 3D set reconstruction, but such methods typically discretize the state space and thus do not scale well to high-dimensional systems. 

With these analogies and challenges in mind, we propose an algorithm which consists of incrementally sampling new parameters $\z$ for which the resulting states lie outside the convex hull of the previous samples. Specifically, we search for new parameters which maximize the metric
\begin{equation}\label{eq:multisteps_adv_cost}
\begin{aligned}
    \mathcal{L}^M(\z) = \frac{1}{N}\sum_{k=1}^N \|\x_k(\z)-\bc_k^M \|_{\bQ_k^M}^2,
    ~\text{with}~\bQ_k^M={\textstyle\big(\frac{1}{M{-}1}\sum_{j=1}^M  \big(\x_k^j{-}\bc_k^M\big)\big(\x_k^j{-}\bc_k^M\big)^T \big)^{-1}},
\end{aligned}
\end{equation}
where $\bc_k^M$ is the geometric center of $\X_k^M$ computed previously with $\smash{\{\x_k^j\}_{j=1}^M}$, $\bQ_k^M$ are positive definite matrices, and $\x_k(\z)=\f(\cdot, \ac_{k-1}, \btheta, \w_{k-1})\circ\dots\circ\f(\x_0, \ac_0, \btheta, \w_0))$. Although we take the mean over multiple time steps, it is also possible to assign different time steps for different samples, but this yields minimal differences in practice. We propose to maximize \eqref{eq:multisteps_adv_cost} using Projected Gradient Ascent (PGA), and outline the algorithm in Algorithm \ref{alg:robUP!}. As the speed of convergence of PGA depends on the condition number of the Hessian of the objective function \citep{BoydVandenberghe2004}, we choose $\bQ_k^M$ as the inverse of the covariance matrix of $\smash{\{\x_k^j\}_{j=1}^M}$, which also provides a simple method to weigh the different dimensions of the state trajectory. The projection step of the parameters onto $\Z{=}\X_0{\times}\U^{\scalebox{.6}{$k$}}{\times}\Theta{\times}\W^{\scalebox{.6}{$k{-}1$}}$ can be efficiently computed for common compact sets used for learning-based control applications (e.g., rectangular and ellipsoidal sets \citep{JiaComparisonEllipsoid2017}).

\vspace{-5mm}\begin{figure}[!htb]
\begin{minipage}{.59\linewidth}
\centering
    \begin{algorithm}[H]
    \small
	\caption{Robust UP using Adversarial Sampling (\textbf{robUP!})} \label{alg:robUP!}
	\vspace{1mm}
	\textbf{Input}: 
	    Sampled parameters $\smash{\{\z^j\}_{j=1}^M{=}(\x_0^j,\ac^j,\btheta^j,\w^j)_{j=1}^M}$
	    \\[1mm]
	\textbf{Parameters}: 
	    Stepsize
	    $\lr$, \ 
	    nb. iters. $n_{\textrm{adv}}$
	    \\[1mm]
	\textbf{Output}: 
	    Particles $\smash{\mathcal{X}_N^{n_\textrm{adv}}{=}\{\x^j\}_{j=1}^{M\cdot n_{\textrm{adv}}}}$ 
	\\[-1mm]
	\hrule\vspace{1mm}
	\begin{algorithmic}[1]
	\State $\x^j_{1:N}\gets  \textrm{Propagate}(\z^j), \quad \ j=1,\mydots,M$
	\State $\mathcal{X}_N^0=\{\x^j_{1:N}\}_{j=1}^M$
	\For{$i=1,\dots,n_{\textrm{adv}}$}
%
		\State $\z^j \gets \z^j + \lr\, \nabla_{\z} \Loss^M(\x^j_{1:N}), \hspace{4.8mm} 
		j=1,\mydots,M$
		\State $\z^j \gets \textrm{Proj}_{\Z}(\z^j), \hspace{20mm} 
		j=1,\mydots,M$
		\State
		 $\x^j_{1:N}\gets  \textrm{Propagate}(\z^j), \hspace{11.5mm}
		 j=1,\mydots,M$
	\State $\mathcal{X}_N^i\gets\mathcal{X}_N^{i-1}\cup\{\x^j_{1:N}\}_{j=1}^M$
	\EndFor
	\State\Return $\textrm{Co}\big(\mathcal{X}_N^{n_\textrm{adv}}\big)$
	\end{algorithmic}
\end{algorithm}
\end{minipage}%
\begin{minipage}{0.04\linewidth}
\centering
\hfill
\end{minipage}%
\begin{minipage}{.38\linewidth}
\centering
\vspace{3.5mm}
\includegraphics[width=0.85\linewidth]{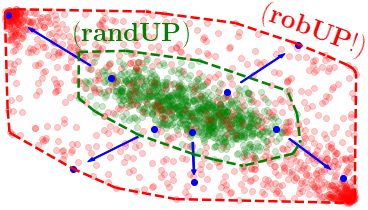}
\label{fig:adv:effect}
\caption{By varying $\x_0$, $\ac$, $\btheta$, and $\w$ along the gradient of $\Loss(\cdot)$ using projected gradient ascent, (\textbf{robUP!}) provides a simple and efficient sampling scheme which robustifies (\textbf{randUP}) 
by finding adversarial parameter samples. 
This figure shows the effect of a single adversarial step on a uncertain system subject to disturbances.}
\end{minipage}
\end{figure}

\vspace{-4mm}
\section{Leveraging System-Specific Properties and Applications}\label{sec:applications}
\vspace{-2mm}

Despite minimal assumptions about the system, we derived an asymptotic convergence guarantee for our method in Theorem~\ref{thm:conv_convHulls_randSets} using random set theory. Stronger guarantees can be derived with additional assumptions. Firstly, if the reachable sets are convex, our method is guaranteed to always provide inner approximations for any finite number of samples, and our tool can be immediately used for falsification \citep{BhatiaIncrem2004}. This is the case for any convex dynamical system and uncertainty sets, e.g., for time-varying linear systems with polytopic uncertainty, or certain neural network architectures \citep{Amos2017,Chen2019ICLR}. Secondly, if smoothness properties of $\f$ are available, inflating our approximate convex hulls by a quantity proportional to the Lipschitz constant and the distance between parameter samples suffices to guarantee an over-approximation. 
Our current approach remains an interesting additional option to efficiently approximate reachable sets whenever it is a challenge to compute tight Lipschitz constants. Thirdly, as discussed in Section \ref{sec:UP_randSets}, we believe that rates of convergence could be derived by leveraging  random set theory and smoothness properties of the boundary of the reachable set. Since current methods make specific assumptions on the distribution of the samples within $\X_k$ \citep{Kor1993_BookMinimaxImgReconstr,Rodriguez2016}, we leave such extensions for future work. Finally, our method can be used to improve the efficiency of other algorithms. Specifically, we present a direct application of our method to the problem of selecting the correct homotopy class of paths for robust planning in the next section.

\vspace{-2mm}
\section{Results and Applications}\label{sec:result}


\vspace{-2mm}\paragraph{Linear system} We first compare our approach with \citep{LiebenweinRSS2018}, which can guarantee $(1\,{-}\,\epsilon)$ volume coverage of the true reachable set. However, this method requires an oracle performing reachability analysis over $\U$ from any state, cannot handle parameters uncertainty and disturbances, requires knowledge of Lipschitz constants, the surface area, and volume of the reachable sets, and requires taking a union of reachable sets, resulting in a non-convex set which is challenging to compute. For these reasons, we provide comparisons on a system for which these assumptions hold, and consider $\x_1\,{=}\,\x_0\,{+}\,\ac_0, \ \x_1\,{\in}\,\R^n, \x_0\,{\in}\,[-1,1]^n$, $\ac_0\,{\in}\, [-\bar{u},\bar{u}]^n$. We set $\eta = 1$ and $n_{\textrm{adv}}\,{=}\,1$, sample uniformly over $\X_0 \times \U_0$, and use a grid to construct a $\delta$-covering of $\X_0$ for \citep{LiebenweinRSS2018}. Figure~\ref{fig:RSS_comparison} reports comparisons of volume coverage. Whereas (\textbf{randUP}) and (\textbf{robUP!}) provide consistent volume coverage across all problems, the convergence rate of \citep{LiebenweinRSS2018} 
is heavily dependent on $\bar{u}$. Furthermore, (\textbf{robUP!}) significantly outperforms other methods with a single adversarial sampling step only. 

\begin{figure}[!htb]
    \centering
    \includegraphics[width=0.295\linewidth]{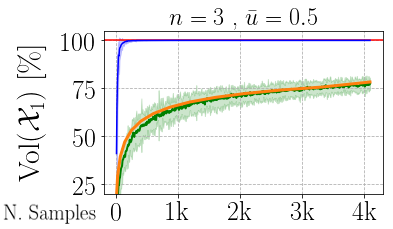}
    \includegraphics[width=0.226\linewidth]{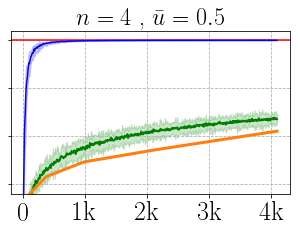}
    \includegraphics[width=0.226\linewidth]{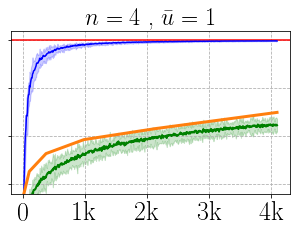}
    \includegraphics[width=0.226\linewidth]{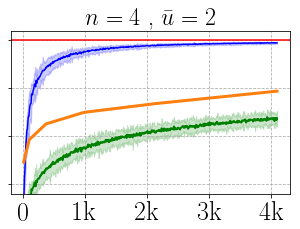}
    \caption{Comparisons on linear system with $\bar{u}\,{\in}\,\{0.5,1,2\}$, $n\,{\in}\,\{3,4\}$, and $3\sigma$ confidence bounds across 10 experiments each. Blue: (\textbf{robUP!}), 
    green: (\textbf{randUP}), 
    orange: \citep{LiebenweinRSS2018}. 
    }
    \label{fig:RSS_comparison}
\end{figure}

\vspace{-5mm}\paragraph{Neural network dynamics} Next, we demonstrate our method on a system with learned dynamics. Specifically, we consider a double integrator with $\x_k = (\pos_k,\vel_k) \in \R^4$, $\ac_k \in \R^2$, $\pos_{k+1} = \pos_k + \vel_k$, and $\vel_{k+1} = \vel_k + \ac_k$. We train a fully connected neural network with two hidden layers of sizes $(128,128)$ and $\tanh(\cdot)$ activation functions. We use a quadratic loss and $\mathcal{L}_2$ regularization, and provide further details about our implementation in PyTorch \citep{paszke2017automatic} in the Appendix. We randomize initial conditions and control trajectories within compact sets $\X_0 \times \U_k^N$ over a horizon of $N=10$ steps, where $\X_0$ are ellipsoidal sets, and $\U_k\,{=}\,\{\ac_k\}$ are fixed open-loop control trajectories. Since the true system is linear, and our trained model achieves an error below $10^{-7}$ over the state space, we can compute the volume  of the true reachable sets. For each of the 100 initial conditions and fixed control trajectories, we evaluate the volume coverage of both (\textbf{randUP}) and (\textbf{robUP!}), and report all results in Figure~\ref{fig:results:nn_lip}. To compute the projection step onto $\X_0$ in (\textbf{robUP!}), we use the Self-adaptive Alternating Direction Method of Multipliers (S-ADMM) described in \citep{JiaComparisonEllipsoid2017}.

\begin{figure}[!htb]
    \captionsetup{justification=justified}
    \begin{minipage}{.281\linewidth}
        \centering
        \includegraphics[width=1\linewidth]{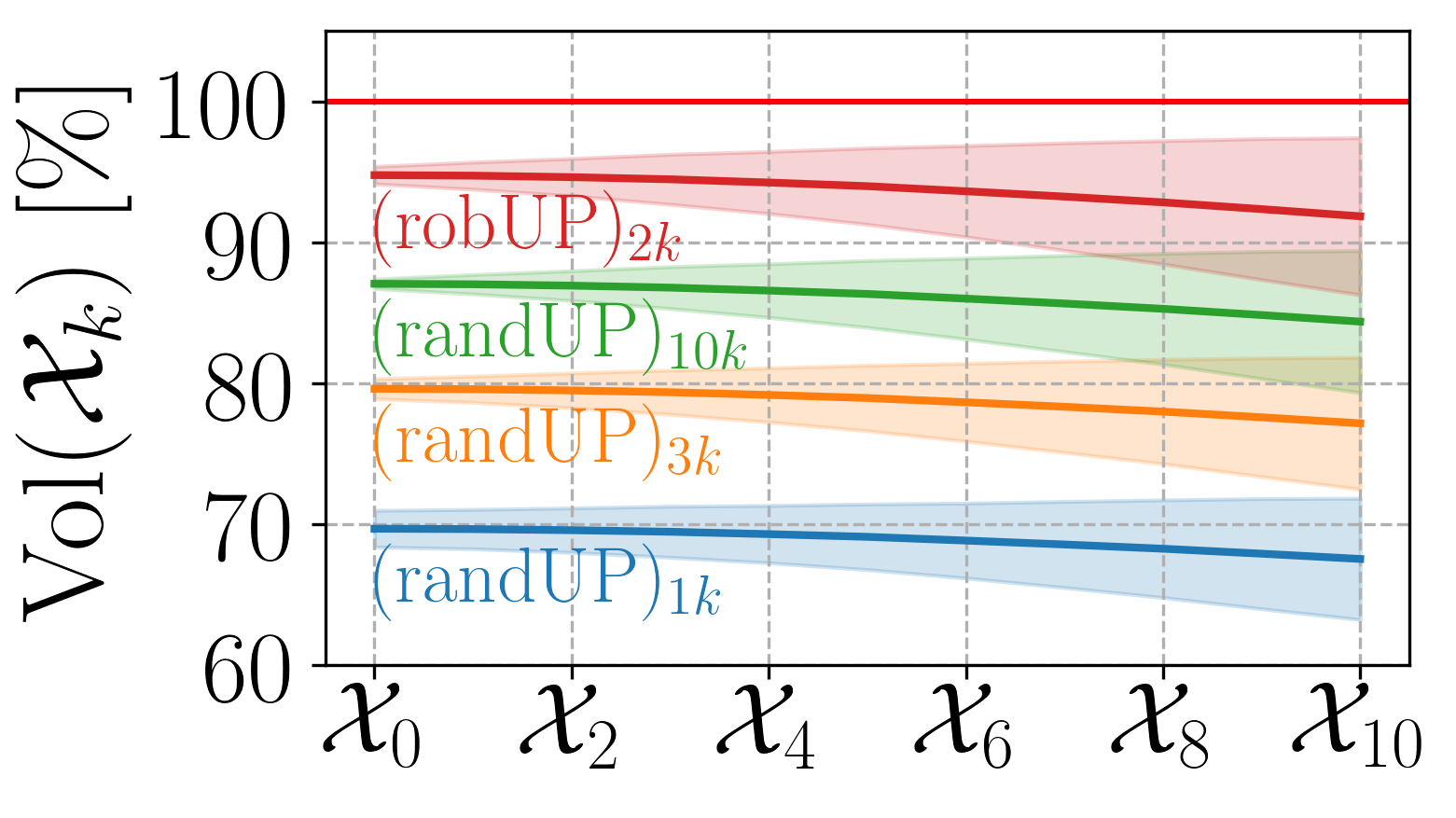}
    \end{minipage}%
    \begin{minipage}{.24\linewidth}
        \centering
        \includegraphics[width=1\linewidth]{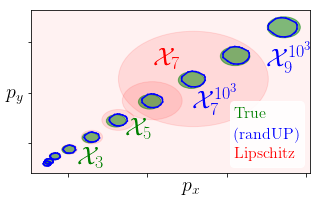}
    \end{minipage}%
    \begin{minipage}{0.452\linewidth}
        \centering\small
        \begin{tabular}{|C|G|G|G|G|G|G|}
            \hline 
            $\%$ of True Vol.   
                &$\X_0$ &$\X_1$ &$\X_2$ &$\X_4$ &$\X_5$\Bstrut 
        \end{tabular}
        \\[-0.3mm] 
        \begin{tabular}{|C|S|S|S|S|S|S|}
            \hline 
            \Tstrut(\textbf{randUP})${}^{\scalebox{.7}{$M{=}3\textrm{k}$}}$
                & 80    & 80    & 79    & 79    & 79\Bstrut 
            \\ 
            \cline{1-6}\Tstrut(\textbf{robUP!})${}^{\scalebox{.7}{$M{=}2\textrm{k}$}}_{\scalebox{.7}{$n_{\textrm{adv}}{=}1$}}$
                & 95    & 95    & 95    & 94    & 94\Bstrut 
            \\ 
            \cline{1-6}\Tstrut Lipschitz \citep{koller2018}
                & 100   & 170   & 321   & 1824  & 6542\Bstrut
            \\ 
            \hline 
        \end{tabular}
    \end{minipage}

    \begin{minipage}{.281\linewidth}
        \centering
        \includegraphics[width=1\linewidth]{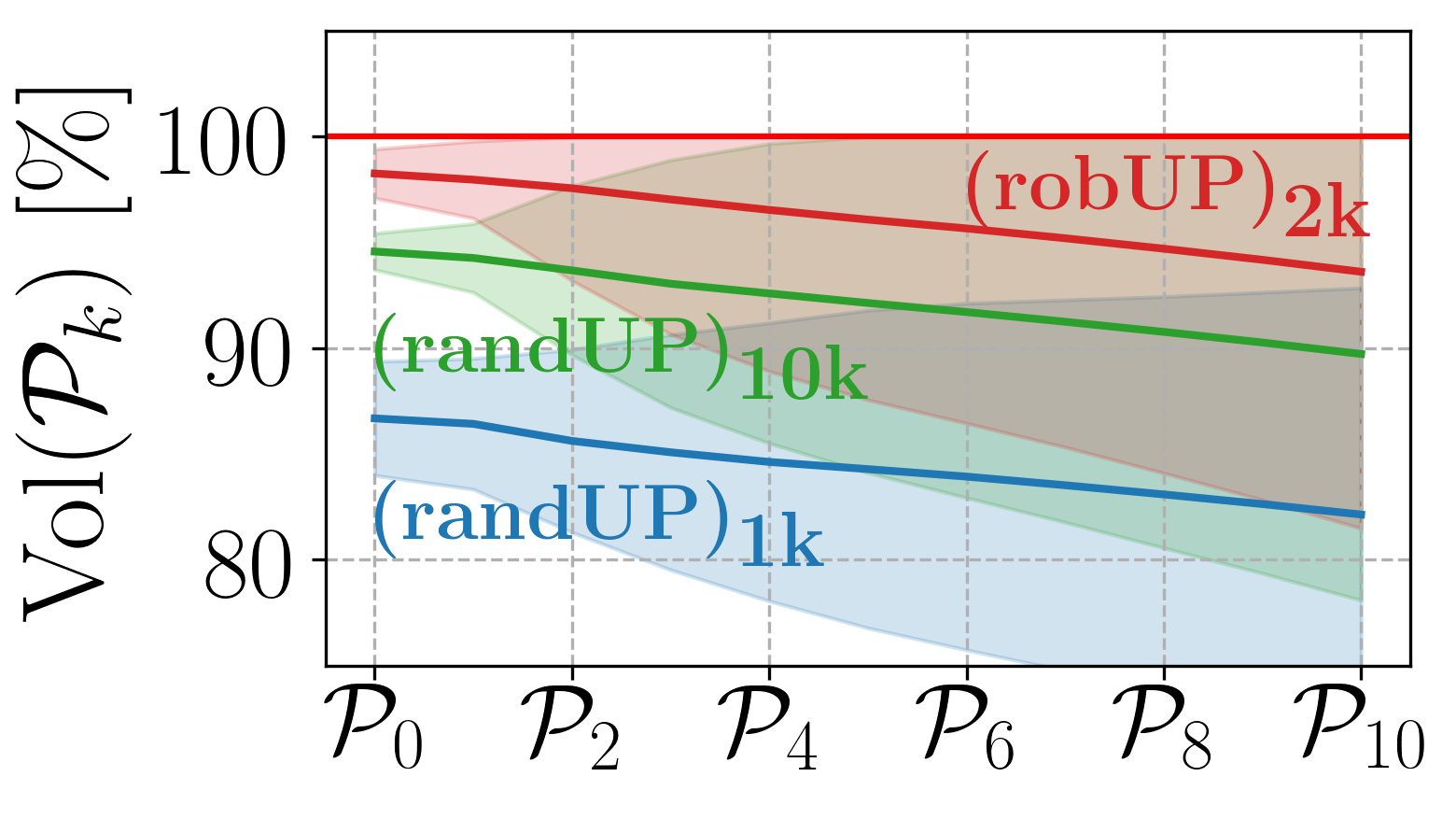}
    \end{minipage}%
    \hspace{1mm}
    \begin{minipage}{.7\linewidth}
        \vspace{2mm}
        \caption{Volume coverage on neural network system for 100 different $\X_0$ and control trajectories. Although volume coverage decreases over time (left, with $\pm 2$ standard dev.), (\textbf{robUP!}) consistently out-performs (\textbf{randUP}). When only considering positions (e.g., for obstacle avoidance), accuracy is better. Comparisons with \citep{koller2018} (top middle, right) show that it is too conservative for a longer horizon, even when the true Lipschitz constant of the system is given.}\label{fig:results:nn_lip}
    \end{minipage}%
    \vspace{-5mm}
\end{figure}


Methods leveraging the Lipschitz constant of the system are capable of conservatively approximating reachable sets. In this experiment, we provide comparisons with the method in~\citep{koller2018}\footnote{In \citep{koller2018}, uncertainty comes in the form of a Gaussian process representing the dynamics, whereas uncertainty is in the initial state in this experiment. 
Since computing tight upper bounds on Lipschitz constants of neural networks is a rapidly evolving field of research, we implement \citep{koller2018} using the true constant of the system.} (see \ref{sec:lipschitz_derivations}), 
which propagates a sequence of ellipsoidal sets that are guaranteed to contain the true system. 
Results in Figure~\ref{fig:results:nn_lip} show that even in situations where the true Lipschitz constant is available, \citep{koller2018} is still too conservative. This is due to the outer-bounding of rectangles with ellipsoids, and to the conservative approximation of the Minkowski sum of ellipsoids, neglecting time correlations of the trajectory on the uncertainty (see Section \ref{sec:problem}). 
Although this method works well for short horizons and low-dimensional systems, 
it is too conservative for longer horizons and high-dimensional nonlinear systems, whereas our approach scales well to more challenging systems, as we show next.

\vspace{-0.5mm}

\begin{minipage}{\linewidth}
	\begin{wrapfigure}{R}{0.20\linewidth}
	\begin{minipage}{1\linewidth}
	\vspace{-9mm}
	\includegraphics[width=0.95\linewidth,trim=0 28 0 0, clip]{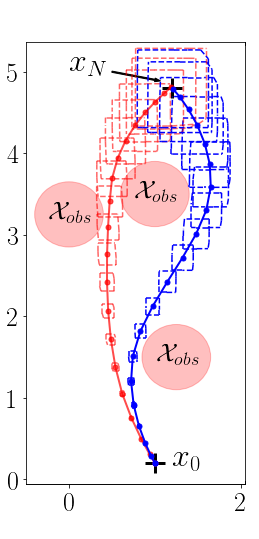}
	\caption{(\textbf{robUP!}) can be used to quickly invalidate unfeasible homotopy classes (red), and provide robust paths for a 13D nonlinear spacecraft system (blue).}\label{fig:spacecraft}
	\vspace{-3mm}
	\end{minipage}%
	\end{wrapfigure}
	
    \vspace{-2mm}\paragraph{Robust planning for a spacecraft} Since our methods have asymptotic convergence guarantees (as $M{\rightarrow}\infty$) to the (conservative) convex hull of the reachable sets, we present an application for the quick selection of feasible homotopy classes for robust planning. Providing feasibility with respect to the convex hull 
    is sufficient, 
    since many optimization-based algorithms linearize constraints to find solutions \citep{LewBonalli2020}. Our approach has connections to algorithms boosting the convergence speed of sampling-based motion planning with reachability analysis \citep{BhatiaIncrem2004,WuHybridRRT2019}, with the difference that our method can account for uncertainty over parameters and disturbances. Specifically, we consider 
    a spacecraft under uncertainty, whose state and control inputs are  
    $\x$ 
    
    ${=}[\bpos,\bv,\bq,\bomega]{\in}\mathbb{R}^{13},\ac{=}[\bF,\bM]{\in}\mathbb{R}^6$, and whose  dynamics are {\scalebox{0.98}{$\dot{\bpos} {=} \bv$}}, 
    {\scalebox{0.98}{$m\dot{\bv} {=} \bF$}}, 
    {\scalebox{0.98}{$\dot{\bq} {=} \frac{1}{2}\boldsymbol\Omega(\bomega)\bq$}}, and {\scalebox{0.98}{$\bJ\dot{\bomega} {=} \bM {-} \bS(\bomega)\bJ\bomega$}}, with {\scalebox{.95}{$J {=} \mathrm{diag}([J_x,J_y,J_z])$}} \citep{LewBonalli2020}. We use a zero-order hold on the controls, an Euler discretization scheme with $\Delta t \,{=}\, 5\mathrm{s}$, and an additive disturbance term $\w_k$ such that {\scalebox{0.99}{$\x_{k{+}1}{=} \f(\x_k,\ac_k,\btheta_k) {+} \w_k$}}. We assume the mass and inertia are unknown with known bounds {\scalebox{.9}{$m \,{\in}\, [7.1,7.3]$}}, {\scalebox{.9}{$J_i\,{\in}\,[0.065,0.075]$}}, {\scalebox{.9}{$|w_{ki}|\,{\leq}\ 10^{-4}$}} for {\scalebox{.9}{$i{=}1,{\mydots},13$}}, and {\scalebox{.9}{$|w_{ki}| \,{\leq}\, 5{\times}10^{-4}$}} for {\scalebox{.9}{$i{=}4,5,6$}}. 
    We perform randomized experiments to choose $M$ and $n_{\textrm{adv}}$, which effectively trade off computation time and accuracy (see the Appendix). 
    We decide to use (\textbf{robUP!}) with $M\,{=}\,100$ and $n_{\textrm{adv}}\,{=}\,1$, as we observe reduced returns over more adversarial steps and samples. We extend the SCP-based planning algorithm from \citep{LewBonalli2020} to leverage our uncertainty propagation schemes, and use outer rectangular confidence sets to reformulate all constraints. Planning results for a problem with three cylindrical obstacles are shown in Figure~\ref{fig:spacecraft}. Solving this problem with $N{=}21$ requires 3~SCP iterations and $875\mathrm{ms}$ on a laptop with an i7-6700 CPU (2.60~GHz) and 8~GB of RAM. If the planner is initialized with a trajectory within the homotopy class on the left, the optimizer does not converge and returns an unfeasible trajectory. This is not the case with a straight-line initialization, in which case convergence is achieved within the correct homotopy class. Although we cannot guarantee that this trajectory is safe for any possible parameters since our guarantees only hold as $M{\rightarrow}\infty$, our approach can be used to rapidly invalidate an unfeasible homotopy class of paths, and thus avoids spending additional computational resources to compute a robust solution to this problem. 
\end{minipage}





\vspace{-2mm}
\section{Conclusion}\label{sec:conclusion}
\vspace{-2.3mm}

Current approaches for reachability analysis have limitations which prevent their use for complex systems and learning-based robotic applications. To fill this gap, this paper presented a simple but theoretically justified sampling-based approach to perform real-time uncertainty propagation. It makes minimal assumptions about the system, can handle general types of uncertainty, scales well to high-dimensional systems, and 
out-performs current approaches.  
As such, it can be combined with conventional uncertainty-aware planners, and used for model-based reinforcement learning with learned dynamics for which no real-time alternatives exist, e.g., Bayesian neural networks \citep{LewSharmaHarrisonPavone2020}.

\vspace{-0.3mm}

Our new interpretation with random set theory opens news avenues for research. Firstly, it could motivate new methods to perform robust neural network training and verification, and design neural network models predicting reachable sets. Secondly, existing results for random set-valued martingales could provide further insights for reachability analysis.  Thirdly, using concentration bounds and smoothness properties of reachable sets could help provide rates of convergence, probabilistic and robust outer-approximations, as long as these assumptions can be verified in practice. Finally, our approach is amenable to parallelization and thus could be further accelerated using GPUs. 


\clearpage


{\small
\textbf{Acknowledgments}
The authors were partially supported by the Office of Naval Research, ONR YIP Program, under Contract N00014-17-1-2433. 
The authors thank the anynomous reviewers for their helpful comments, Spencer Richards for his feedback and suggestions, Robert Dyro for his implementation of (S-ADMM), Riccardo Bonalli for his feedback on theoretical results, James Harrison and Apoorva Sharma for helpful discussions on learning-based control, and Paul-Edouard Sarlin for suggesting related work on shape reconstruction.}


\bibliography{ASL_papers,main}  

\newcommand{\noopsort}[1]{} \newcommand{\printfirst}[2]{#1}
  \newcommand{\singleletter}[1]{#1} \newcommand{\switchargs}[2]{#2#1}
\begin{thebibliography}{50}
\providecommand{\natexlab}[1]{#1}
\providecommand{\url}[1]{\texttt{#1}}
\expandafter\ifx\csname urlstyle\endcsname\relax
  \providecommand{\doi}[1]{doi: #1}\else
  \providecommand{\doi}{doi: \begingroup \urlstyle{rm}\Url}\fi

\bibitem[Abbasi-Yadkori et~al.(2011)Abbasi-Yadkori, P{\'a}l, and
  Szepesv{\'a}ri]{abbasi2011improved}
Y.~Abbasi-Yadkori, D.~P{\'a}l, and C.~Szepesv{\'a}ri.
\newblock Improved algorithms for linear stochastic bandits.
\newblock In \emph{{Conf.\ on Neural Information Processing Systems}}, 2011.

\bibitem[Molchanov(2017)]{Molchanov_BookTheoryOfRandomSets2017}
I.~Molchanov.
\newblock \emph{Theory of Random Sets}.
\newblock {Springer-Verlag}, second edition, 2017.

\bibitem[Schmerling and Pavone(2017)]{SchmerlingPavone2017}
E.~Schmerling and M.~Pavone.
\newblock Evaluating trajectory collision probability through adaptive
  importance sampling for safe motion planning.
\newblock In \emph{{Robotics: Science and Systems}}, 2017.

\bibitem[Loquercio et~al.(2020)Loquercio, Segu, and
  Scaramuzza]{LoquercioRAL2020}
A.~Loquercio, M.~Segu, and D.~Scaramuzza.
\newblock A general framework for uncertainty estimation in deep learning.
\newblock \emph{{IEEE Robotics and Automation Letters}}, 5\penalty0
  (2):\penalty0 3153--3160, 2020.

\bibitem[Devonport and Arcak(2020)]{EstReachSets_L4DC20}
A.~Devonport and M.~Arcak.
\newblock Estimating reachable sets with scenario optimization.
\newblock In \emph{2nd Annual Conference on Learning for Dynamics \& Control},
  2020.

\bibitem[Esfahani et~al.(2015)Esfahani, Sutter, and
  Lygeros]{EsfahaniScenario2015}
M.~Esfahani, P.~T. Sutter, and J.~Lygeros.
\newblock Performance bounds for the scenario approach and an extension to a
  class of non-convex programs.
\newblock \emph{{IEEE Transactions on Automatic Control}}, 60\penalty0 (1),
  2015.

\bibitem[Hewing et~al.(2018)Hewing, Kabzan, and Zeilinger]{hewing2018cautious}
L.~Hewing, J.~Kabzan, and M.~N. Zeilinger.
\newblock Cautious model predictive control using {Gaussian} process
  regression.
\newblock \emph{{IEEE Transactions on Control Systems Technology}}, 2018.

\bibitem[Lew et~al.(2020)Lew, Bonalli, and Pavone]{LewBonalli2020}
T.~Lew, R.~Bonalli, and M.~Pavone.
\newblock Chance-constrained sequential convex programming for robust
  trajectory optimization.
\newblock In \emph{{European Control Conference}}, 2020.

\bibitem[Castillo-Lopez et~al.(2019)Castillo-Lopez, Ludivig, Sanchez-Lopez,
  Olivares-Mendez, and Voos]{CastilloRAL2020}
M.~Castillo-Lopez, S.~A. Ludivig, P. Sajadi-Alamdari, J.~L. Sanchez-Lopez,
  M.~A. Olivares-Mendez, and H.~Voos.
\newblock A real-time approach for chance-constrained motion planning with
  dynamic obstacles.
\newblock \emph{{IEEE Robotics and Automation Letters}}, 5\penalty0
  (2):\penalty0 3620 -- 3625, 2019.

\bibitem[Berret and Jean(2020)]{BerretAutomatica2020}
B.~Berret and F.~Jean.
\newblock Efficient computation of optimal open-loop controls for stochastic
  systems.
\newblock \emph{{IEEE Robotics and Automation Letters}}, 115\penalty0 (1),
  2020.

\bibitem[Liebenwein et~al.(2018)Liebenwein, Baykal, Gilitschenski, Karaman, and
  Rus]{LiebenweinRSS2018}
L.~Liebenwein, C.~Baykal, I.~Gilitschenski, S.~Karaman, and D.~Rus.
\newblock Sampling-based approximation algorithms for reachability analysis
  with provable guarantees.
\newblock In \emph{{Robotics: Science and Systems}}, 2018.

\bibitem[Korostelev and Tsybakov(1993)]{Kor1993_BookMinimaxImgReconstr}
A.~P. Korostelev and A.~B. Tsybakov.
\newblock \emph{Minimax Theory of Image Reconstruction}.
\newblock {Springer-Verlag}, 1 edition, 1993.

\bibitem[Koller et~al.(2018)Koller, Berkenkamp, Turchetta, and
  Krause]{koller2018}
T.~Koller, F.~Berkenkamp, M.~Turchetta, and A.~Krause.
\newblock Learning-based model predictive control for safe exploration.
\newblock In \emph{{Proc.\ IEEE Conf.\ on Decision and Control}}, 2018.

\bibitem[Dean et~al.(2020)Dean, Matni, Recht, and Ye]{Dean2019RobustGF}
S.~Dean, N.~Matni, B.~Recht, and V.~Ye.
\newblock Robust guarantees for perception-based control.
\newblock In \emph{2nd Annual Conference on Learning for Dynamics \& Control},
  2020.

\bibitem[Fazlyab et~al.(2019)Fazlyab, Robey, Hassani, Morari, and
  Pappas]{fazlyab2019efficient}
M.~Fazlyab, A.~Robey, H.~Hassani, M.~Morari, and G.~J. Pappas.
\newblock Efficient and accurate estimation of lipschitz constants for deep
  neural networks.
\newblock In \emph{{Conf.\ on Neural Information Processing Systems}}, 2019.

\bibitem[Latorre et~al.(2020)Latorre, Rolland, and Cevher]{lipsPoly2020}
F.~Latorre, P.~Rolland, and V.~Cevher.
\newblock Lipschitz constant estimation of neural networks via sparse
  polynomial optimization.
\newblock In \emph{{Int.\ Conf.\ on Machine Learning}}, 2020.

\bibitem[Finn et~al.(2017)Finn, Abbeel, and Levine]{MAML2017}
C.~Finn, P.~Abbeel, and S.~Levine.
\newblock Model-agnostic meta-learning for fast adaptation of deep networks.
\newblock In \emph{{Int.\ Conf.\ on Machine Learning}}, 2017.

\bibitem[Harrison et~al.(2018)Harrison, Sharma, and
  Pavone]{HarrisonSharmaEtAl2018}
J.~Harrison, A.~Sharma, and M.~Pavone.
\newblock Meta-learning priors for efficient online bayesian regression.
\newblock In \emph{{Workshop on Algorithmic Foundations of Robotics}}, 2018.

\bibitem[Lew et~al.(2020)Lew, Sharma, Harrison, and
  Pavone]{LewSharmaHarrisonPavone2020}
T.~Lew, A.~Sharma, J.~Harrison, and M.~Pavone.
\newblock Safe model-based meta-reinforcement learning: A sequential
  exploration-exploitation framework, 2020.
\newblock Available at \url{https://arxiv.org/abs/2008.11700}.

\bibitem[Weng et~al.(2018)Weng, Zhang, Yi, Su, Gao, Hsieh, and
  Daniel]{WengExtreme2018}
T.~Weng, P.~Zhang, H.and~Chen, J.~Yi, D.~Su, Y.~Gao, C.~Hsieh, and L.~Daniel.
\newblock Evaluating the robustness of neural networks: an extreme value theory
  approach.
\newblock In \emph{{Int.\ Conf.\ on Learning Representations}}, 2018.

\bibitem[Bansal et~al.(2017)Bansal, Chen, and Tomlin]{HJIoverviewBansal2017}
S.~Bansal, S.~L. Chen, M.~Herbert, and C.~J. Tomlin.
\newblock {Hamilton}-{Jacobi} reachability: A brief overview and recent
  advances.
\newblock In \emph{{Proc.\ IEEE Conf.\ on Decision and Control}}, 2017.

\bibitem[Fisac et~al.(2018)Fisac, Bajcsy, Herbert, Fridovich-Keil, Wang,
  Tomlin, and Dragan]{FisacRSS2018}
J.~F. Fisac, A.~Bajcsy, S.~L. Herbert, D.~Fridovich-Keil, S.~Wang, C.~J.
  Tomlin, and A.~D. Dragan.
\newblock Probabilistically safe robot planning with confidence-based human
  predictions.
\newblock In \emph{{Robotics: Science and Systems}}, 2018.

\bibitem[Rubies-Royo et~al.(2019)Rubies-Royo, Fridovich-Keil, Herbert, and
  Tomlin]{RubiesICRA2019}
V.~Rubies-Royo, D.~Fridovich-Keil, S.~L. Herbert, and C.~J. Tomlin.
\newblock A classification-based approach for approximate reachability.
\newblock In \emph{{Proc.\ IEEE Conf.\ on Robotics and Automation}}, 2019.

\bibitem[Fan et~al.(2020)Fan, Agha-mohammadi, and Theodorou]{fan2020deep}
D.~D. Fan, A.~Agha-mohammadi, and E.~A. Theodorou.
\newblock Deep learning tubes for tube {MPC}.
\newblock In \emph{{Robotics: Science and Systems}}, 2020.

\bibitem[Chen et~al.(2012)Chen, \'Abrah\'am, and
  Sankaranarayanan]{ChenFlowstar2012}
X.~Chen, E.~\'Abrah\'am, and S.~Sankaranarayanan.
\newblock Taylor model flowpipe construction for non-linear hybrid systems.
\newblock In \emph{Proc.\ of IEEE Real-Time Systems Symposium}, 2012.

\bibitem[Ivanov et~al.(2019)Ivanov, Weimer, Alur, Pappas, and
  Lee]{IvanovVerisig2019}
R.~Ivanov, J.~Weimer, R.~Alur, G.~J. Pappas, and I.~Lee.
\newblock Verisig: verifying safety properties of hybrid systems with neural
  network controllers.
\newblock In \emph{{Hybrid Systems: Computation and Control}}, 2019.

\bibitem[Herbert et~al.(2019)Herbert, Bansal, and
  Tomlin]{warmStartingHJI_Herbert2019}
S.~Herbert, S.~L.~Ghosh, S.~Bansal, and C.~J. Tomlin.
\newblock Reachability-based safety guarantees using efficient initializations.
\newblock In \emph{{Proc.\ IEEE Conf.\ on Decision and Control}}, 2019.

\bibitem[Rosolia and Borrelli(2019)]{rosolia2019}
U.~Rosolia and F.~Borrelli.
\newblock Sample-based learning model predictive control for linear uncertain
  systems, 2019.
\newblock Available at \url{https://arxiv.org/abs/1904.06432}.

\bibitem[Hewing et~al.(2020)Hewing, Arcari, Frohlich, and
  Zeilinger]{simGPs_L4DC20}
L.~Hewing, E.~Arcari, L.~P. Frohlich, and M.~N. Zeilinger.
\newblock On simulation and trajectory prediction with {Gaussian} process
  dynamics.
\newblock In \emph{2nd Annual Conference on Learning for Dynamics \& Control},
  2020.

\bibitem[Ripley and Rasson(1977)]{RipleyPoissonForest1977}
B.~D. Ripley and J.~P. Rasson.
\newblock Finding the edge of a poisson forest.
\newblock \emph{{Journal of Applied Probability}}, 14:\penalty0 483--491, 1977.

\bibitem[Rodriguez-Casal and Saavedra-Nieves(2016)]{Rodriguez2016}
A.~Rodriguez-Casal and P.~Saavedra-Nieves.
\newblock A fully data-driven method for estimating the shape of a point cloud.
\newblock \emph{ESAIM: Probability and Statistics}, 20\penalty0 (1):\penalty0
  332--348, 2016.

\bibitem[Büeler et~al.(2000)Büeler, Enge, and Fukuda]{Bueler2000}
B.~Büeler, A.~Enge, and K.~Fukuda.
\newblock Exact volume computation for polytopes: A practical study.
\newblock In \emph{Polytopes - Combinatorics and Computation}, pages 131--154.
  2000.

\bibitem[Chevalier et~al.(2014)Chevalier, Bect, Ginsbourger, Vazquez, Picheny,
  and Richet]{ChevalierKriging2014}
C.~Chevalier, J.~Bect, D.~Ginsbourger, E.~Vazquez, V.~Picheny, and Y.~Richet.
\newblock Fast parallel kriging-based stepwise uncertainty reduction with
  application to the identification of an excursion set.
\newblock \emph{Technometrics}, 56\penalty0 (4):\penalty0 455--465, 2014.

\bibitem[Goodfellow et~al.(2015)Goodfellow, Shlens, and
  Szegedy]{GoodfellowAdv2015}
I.~J. Goodfellow, J.~Shlens, and C.~Szegedy.
\newblock Explaining and harnessing adversarial examples.
\newblock In \emph{{Int.\ Conf.\ on Learning Representations}}, 2015.

\bibitem[Kurakin et~al.(2017)Kurakin, Goodfellow, and Bengio]{Kurakin2017}
A.~Kurakin, I.~J. Goodfellow, and S.~Bengio.
\newblock Adversarial examples in the physical world, 2017.
\newblock Available at \url{https://arxiv.org/abs/1607.02533}.

\bibitem[Dong et~al.(2018)Dong, Liao, Pang, Su, Zhu, Hu, and
  Li]{dong2018boosting}
Y.~Dong, F.~Liao, T.~Pang, H.~Su, J.~Zhu, X.~Hu, and J.~Li.
\newblock Boosting adversarial attacks with momentum.
\newblock In \emph{{IEEE Conf.\ on Computer Vision and Pattern Recognition}},
  2018.

\bibitem[Fan et~al.(2020)Fan, Su, and Guibas]{Fan2016}
H.~Fan, H.~Su, and L.~Guibas.
\newblock A point set generation network for {3D} object reconstruction from a
  single image, 2020.
\newblock Available at \url{https://arxiv.org/abs/1612.00603}.

\bibitem[Hoppe et~al.(1992)Hoppe, DeRose, Duchamp, McDonald, and
  Stuetzle]{Hugues1992}
H.~Hoppe, T.~DeRose, T.~Duchamp, J.~McDonald, and W.~Stuetzle.
\newblock Surface reconstruction from unorganized points.
\newblock In \emph{{ACM Proc.\ of the Annual Conf.\ on Computer Graphics and
  Interactive Techniques}}, 1992.

\bibitem[Cohen-Or et~al.(1998)Cohen-Or, Levin, and Solomovici]{Cohen1998}
D.~Cohen-Or, D.~Levin, and A.~Solomovici.
\newblock Three-dimensional distance field metamorphosis.
\newblock \emph{{ACM Transactions on Graphics}}, 17\penalty0 (2):\penalty0
  116--141, 1998.

\bibitem[Chabra et~al.(2020)Chabra, Lenssen, Ilg, Schmidt, Straub, Lovegrove,
  and Newcombe]{Chabra2020}
R.~Chabra, J.~E. Lenssen, E.~Ilg, T.~Schmidt, J.~Straub, S.~Lovegrove, and
  R.~Newcombe.
\newblock Deep local shapes: Learning local {SDF} priors for detailed {3D}
  reconstruction, 2020.
\newblock Available at \url{https://arxiv.org/abs/2003.10983}.

\bibitem[Boyd and Vandenberghe(2004)]{BoydVandenberghe2004}
S.~Boyd and L.~Vandenberghe.
\newblock \emph{Convex optimization}.
\newblock {Cambridge Univ.\ Press}, 2004.

\bibitem[Jia et~al.(2017)Jia, Cai, and Han]{JiaComparisonEllipsoid2017}
Z.~Jia, X.~Cai, and D.~Han.
\newblock Comparison of several fast algorithms for projection onto an
  ellipsoid.
\newblock \emph{{Journal of Computational and Applied Mathematics}},
  319\penalty0 (1):\penalty0 320--337, 2017.

\bibitem[Bhatia and Frazzoli(2004)]{BhatiaIncrem2004}
A.~Bhatia and E.~Frazzoli.
\newblock Incremental search methods for reachability analysis of continuous
  and hybrid systems.
\newblock In \emph{{Hybrid Systems: Computation and Control}}, pages 142--156,
  Berlin, Heidelberg, 2004. Springer Berlin Heidelberg.

\bibitem[Amos et~al.(2017)Amos, Xu, and Kolter]{Amos2017}
B.~Amos, L.~Xu, and Z.~Kolter.
\newblock Input convex neural networks.
\newblock In \emph{{Int.\ Conf.\ on Machine Learning}}, 2017.

\bibitem[Chen et~al.(2019)Chen, Shi, and Zhang]{Chen2019ICLR}
Y.~Chen, Y.~Shi, and B.~Zhang.
\newblock Optimal control via neural networks: a convex approach.
\newblock In \emph{{Int.\ Conf.\ on Learning Representations}}, 2019.

\bibitem[Paszke et~al.(2017)Paszke, Gross, Chintala, Chanan, Yang, DeVito, Lin,
  Desmaison, Antiga, and Lerer]{paszke2017automatic}
A.~Paszke, D.~Gross, S.~Chintala, G.~Chanan, E.~Yang, Z.~DeVito, Z.~Lin,
  A.~Desmaison, L.~Antiga, and A.~Lerer.
\newblock Automatic differentiation in pytorch.
\newblock In \emph{{Conf.\ on Neural Information Processing Systems}}, 2017.

\bibitem[Wu et~al.(2019)Wu, Sadraddini, and Tedrake]{WuHybridRRT2019}
A.~Wu, S.~Sadraddini, and R.~Tedrake.
\newblock {R3T}: Rapidly-exploring random reachable set tree for optimal
  kinodynamic planning of nonlinear hybrid systems, 2019.
\newblock Available at
  \url{https://groups.csail.mit.edu/robotics-center/public_papers/Wu20.pdf}.

\bibitem[Kingma and Ba(2015)]{AdamPyTorch}
D.~P. Kingma and J.~L. Ba.
\newblock Adam: A method for stochastic optimization.
\newblock In \emph{{Int.\ Conf.\ on Learning Representations}}, 2015.

\bibitem[Sun and Freund(2004)]{Sun2004}
P.~Sun and R.~M. Freund.
\newblock Computation of minimum-volume covering ellipsoids.
\newblock \emph{{Operations Research}}, 52\penalty0 (5), 2004.

\bibitem[Stellato et~al.(2017)Stellato, Banjac, Goulart, Bemporad, and
  Boyd]{StellatoBanjacEtAl2017}
B.~Stellato, G.~Banjac, P.~Goulart, A.~Bemporad, and S.~Boyd.
\newblock {OSQP}: An operator splitting solver for quadratic programs.
\newblock 2017.
\newblock {Available at } \url{https://arxiv.org/abs/1711.08013}.

\end{thebibliography}

\appendix


\section{Proof of Theorem \ref{thm:conv_convHulls_randSets}}

For ease of reading, we first restate Theorem \ref{thm:conv_convHulls_randSets}:

\setcounter{theorem}{1}
\begin{theorem}[Convergence of Random Convex Hulls using (\textbf{randUP})\vspace{2mm}] 
	Let 
	$\smash{\{(\x_0^j,\ac^j,\btheta^j,\w^j)\}_{j=1}^m}$
	 be i.i.d. sampled parameters in $\X_0{\times}\U^{\scalebox{.6}{$k{-}1$}}{\times}\Theta{\times}\W^{\scalebox{.6}{$k{-}1$}}$.  
	Define $\smash{\x_k^j}$  
	according to 
	\eqref{eq:dyns}, 
	and denote the resulting convex hulls as $\smash{\X_k^m \,{=}\, \textrm{Co}(\{\x_k^j\}_{j=1}^m)}$.
	Assume that the sampling distribution of the parameters satisfies $\smash{\Prob(\x_k^j\,{\in}\, G_k)\,{>}\,0}$ for any open set $G_k$ s.t. $\X_k\,{\cap}\,G_k\,{\neq}\, \emptyset$. 
	
	Then, as $m\rightarrow\infty$, $\X_k^m$ converges to the convex hull of the reachable set $\textrm{Co}(\X_k)$ almost surely. 
\end{theorem}
%
%

\begin{proof} 
Without loss of generality,  
we prove this theorem for any arbitrary fixed time index $k$. 
For conciseness, 
we drop the index $k$ and denote 
$\smash{(\x^j,\X^m,\X)} \,{=}\, \smash{(\x^j_k,\X^m_k,\textrm{Co}(\X_k))}$, 
the sampled parameters tuple $\z{=}(\x_0,\ac,\btheta,\w)\,{\in}\,\R^z$, and 
the compact parameters set $\Z{=}\X_0{\times}\U^{\scalebox{.6}{$k$}}{\times}\Theta{\times}\W^{\scalebox{.6}{$k{-}1$}}$. 
Also, let 
$\f(\z){=}\f(\boldsymbol\cdot,\ac_{k-1},\btheta,\w_{k-1}){\circ}{\dots}{\circ}\f(\x_0,\ac_0,\btheta,\w_{0})$, which is also continuous in $\z$. 

Let $(\Omega,\mathcal{A}, \mathbb{P})$ a probability space, and 
$\mathcal{F}$ the family of closed sets in $\R^n$. 
Define  $\X^m:\Omega\rightarrow\mathcal{F}$,  
with $\smash{\X^m(\omega) = \textrm{Co}(\{\x^j\}_{j=1}^m)}$. 
	Then, $\{\X^m(\omega), m\geq 1\}$ is a sequence of compact random sets satisfying $\X^1(\omega)\subseteq \X^2(\omega) \subseteq \dots $ almost surely (a.s.). 

Indeed, for all $j\,{=}\,1,\dots,m$, $\{\x^j\}$ is a random closed set \citep{Molchanov_BookTheoryOfRandomSets2017}. 
Then, applying \citep[Theorem 1.3.25, (i), (iv)]{Molchanov_BookTheoryOfRandomSets2017}, we obtain that $\X^m=\smash{\textrm{Co}(\{\x^j\}_{j=1}^m)}$ is a random closed set. 
Hence $\{\X^m(\omega), m\geq 1\}$ is a sequence of random closed sets satisfying Def. \ref{def:rand_set}. 
Further, each $\X^m(\omega)$ is compact since it is the convex hull of bounded $\smash{\x^j}$'s, since $\f(\cdot)$ is continuous and $\Z$ is compact. 
Finally, 
by definition, 
$\smash{\X^m} \subseteq\X^{m+1}$ a.s. for all $m$, 
hence $\smash{\X^1(\omega)\subseteq \X^2(\omega)} \subseteq \dots $ a.s.

Now that we proved that our convex hulls are random sets, we proceed with the proof that the sequence $\{\X^m(\omega), m\geq 1\}$ converges to $\X$ 
 by proving that we satisfy (\textbf{C1}) and (\textbf{C2}) of Theorem \ref{thm:conv_randSets_detLim}. 

(\textbf{C1}): Let $K\in \mathcal{K}$ satisfy $\X\cap K = \emptyset$. 
Since $\x^j(\omega)\in\X$ almost surely for all $j$, $\X^m\subset\textrm{Co}(\X)$ almost surely for all $m$. 
	This implies that $\X^m\cap K=\emptyset$ almost surely for all $m$. Thus, $\Prob(\X^m\cap K\neq\emptyset)=0$ and $\sum_{m=1}^\infty \Prob(\X^m\cap K\neq\emptyset)=0$. By the first Borel-Cantelli lemma, $\Prob(\X^m\cap K \neq \emptyset \ \ i.o.) = 0$. 

(\textbf{C2}): Let $G\in \mathcal{G}$ satisfy $\X\cap G \neq \emptyset$. 
To prove that $\Prob(\X^m\cap G  = \emptyset \ \ i.o.) = 0$, we proceed in three steps: (1) we show that sampling points $\x^j$ within $G$ occurs infinitely often (i.o.); 
(2) we use the growth property $\X^m\subseteq\X^{m+1}$ to rewrite (\textbf{C2});
(3) we relate the probability of sampling $\x^j$ within $G$ with the probability of $\X^m$ to intersect with $G$.

(1) 
The goal is to show $\Prob (\x^j\in G \ \ i.o.) = 1$. 

We first note that $\{\omega \,|\, \x^j\in G\} \,{=}\, \{\omega \,|\, \z^j \ \text{s.t.} \ \f(\z^j) \,{\in}\, G\}$. 
Since the parameters $\smash{\z^j}$ are sampled independently for each $j$, 
it can be shown that $\{\omega\,|\,\x^j\in G\}$ are also independent for $j=1,\mydots,m$. 

Then, let $\Z_G\,{\subset}\,\R^z$ such that for all $\z\in\Z_G$, we have $\f(\z)\in G$. 
By continuity of $\f(\cdot)$, $\Z_G$ is also an open set. 
Then, since $\X\cap G \neq \emptyset$ and by assumption on the sampling distribution over parameters, we have  $\Prob(\x^j\in G) = \Prob(\z^j\in\Z_G)=\alpha_G >0$, where $\alpha_G$ does not depend on $j$ since the parameters $(\x_0,\ac,\btheta,\w)$ are sampled i.i.d.. 
Therefore, we obtain $\smash{\sum_{j=1}^\infty\Prob(\x^j\in G)=\infty}$.

Next, since 
the events $\{\omega\,|\,\x^j\,{\in}\, G\}$ are independent, and by the above, 
we apply the second Borel-Cantelli lemma to obtain 
$\smash{\Prob (\x^j\,{\in}\, G \ \ i.o.) \,{=}\, \Prob\big(\cup_{n=1}^\infty\cap_{m=n}^\infty \x^j\,{\in}\, G\big) \,{=}\, 1}$.

From this result and $\cap_{n=1}^\infty A_n \subseteq A_1$, we obtain that 
\begin{equation}\label{eq:prob_union_xm_in_G_equals_1}
\Prob (\x^j\in G \ \ i.o.) \leq \Prob\bigg(\bigcup_{m=1}^\infty\x^m\in G\bigg) 
\implies 
\Prob\bigg(\bigcup_{m=1}^\infty\x^m\in G\bigg) = 1.
\end{equation}

(2) 
Second, we rewrite (\textbf{C2}) as follows:
\begin{align}\nonumber
\Prob(\X^m\cap G = \emptyset \ \ i.o.)
&=
\Prob\bigg(\bigcap_{n=1}^\infty\bigcup_{m=n}^\infty \X^m\cap G = \emptyset\bigg)
=
1 - \Prob\bigg(\bigcup_{n=1}^\infty\bigcap_{m=n}^\infty \X^m\cap G \neq \emptyset\bigg).
\end{align}
Since $\X^m\subseteq\X^{m+1}, \forall m$, we have that $\{\omega \, |\, \bigcap_{m=n}^\infty \X^m\cap G \neq \emptyset\}=\{\omega \, |\, \X^n\cap G \neq \emptyset\}$, and 
\begin{equation}\nonumber
\Prob\bigg(\bigcup_{n=1}^\infty\bigcap_{m=n}^\infty \X^m\cap G \neq \emptyset\bigg)=\Prob\bigg(\bigcup_{n=1}^\infty\X^n\cap G \neq \emptyset\bigg).
\end{equation}
\vspace{-4mm}
\begin{flalign}\label{eq:io_union_iff}
\text{Therefore,} 
&&\Prob(\X^m\cap G = \emptyset \ \ i.o.) = 0 
	\iff 
\Prob\bigg(\bigcup_{n=1}^\infty\X^n\cap G \neq \emptyset\bigg) = 1.
&&
\end{flalign}

\vspace{-4mm}
(3) 
Finally, we combine the two results from above. First, we note that since $\x^j\in\X^j$,
we have $\{\omega \,|\, \x^j(\omega) \in G\} \subseteq
\{\omega \,|\, \X^j(\omega) \cap G \neq \emptyset\}$. 
Hence, we obtain
\vspace{-2.3mm}
\begin{equation*}
\bigcup_{j=1}^\infty\{\omega \,|\, \x^j(\omega) \in G\} \subseteq
\bigcup_{j=1}^\infty\{\omega \,|\, \X^j(\omega) \cap G \neq \emptyset\}.
\end{equation*}
Combining with \eqref{eq:prob_union_xm_in_G_equals_1}, we obtain:
\vspace{-9.3mm}
\begin{equation}\label{eq:prob_union_Xm_intersects_G_equals_1}
\Prob\bigg(\bigcup_{n=1}^\infty\X^n\cap G \neq \emptyset\bigg) = 1.
\end{equation}
Using \eqref{eq:io_union_iff}, \eqref{eq:prob_union_Xm_intersects_G_equals_1} is equivalent to $\Prob(\X^m\cap G = \emptyset \ \ i.o.) = 0$, which conludes the proof of (\textbf{C2}). 

By Theorem \ref{thm:conv_randSets_detLim}, we conclude that the sequence $\{\X^m(\omega),m\geq 1\}$ almost surely converges to the deterministic set $\X$ as $m\rightarrow\infty$. 
As $\X$ is defined as the convex hull of the true reachable set $\textrm{Co}(\X_k)$, this concludes our proof of Theorem \ref{thm:conv_convHulls_randSets}.
%
\end{proof}

\newpage
\section{Further Details and Applications of Adversarial Sampling}

We present further results for adversarial sampling for $n_{\textrm{adv}}\geq 1$ in Figure \ref{fig:adv:freeflyer:adv_effect_iters_all}. 
First, we vizualize the effect of adversarial sampling on a planar version of the 13-dimensional nonlinear spacecraft system under uncertainty presented in Section \ref{sec:result}. This result shows that taking more adversarial steps leads to samples concentrated at the boundaries of the reachable sets. However, this does not necessarily correlate with larger sets, since the convex hull over the samples is taken.

\begin{figure}[!htb]
\centering
\includegraphics[width=1\linewidth]{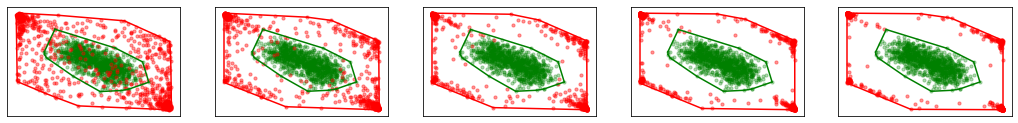}
\caption{Green: (\textbf{randUP}). Red: (\textbf{robUP!}). Effect of adversarial sampling for $n_{\textrm{adv}}\,{\in}\,\{1,\mydots,5\}$ for a planar spacecraft system subject to uncertainty. Projection onto positions given a sequence of open-loop controls. }
\label{fig:adv:freeflyer:adv_effect_iters_all}
\end{figure}

\begin{minipage}{\linewidth}
	\begin{wrapfigure}{R}{0.35\linewidth}
	\begin{minipage}{1\linewidth}
	\vspace{-5mm}
    \centering
	\includegraphics[width=0.9\linewidth]{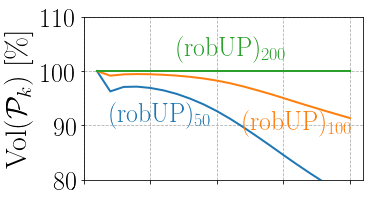}
	\includegraphics[width=0.9\linewidth]{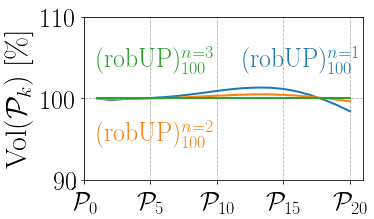}
	\caption{To perform robust trajectory optimization for a spacecraft under uncertainty, we run 500 randomized experiments to choose $M$ and $n_{\textrm{adv}}$. 
	For an horizon $N\,{=}\,20$, executing $(\textbf{robUP!})$ with $n_{\textrm{adv}}\,{=}\,1$ and $M{\in}\{50,100,200\}$ requires an average of $\{83,173, 304\}$ms, respectively, on a laptop with an i7-6700 CPU (2.60~GHz) and 8~GB of RAM. 
}\label{fig:spacecraft_tuning}
	\vspace{-3mm}
	\end{minipage}%
	\end{wrapfigure}
	
    \vspace{-2mm}
    Next, we justify the choice of $M=100$ and $n_{\textrm{adv}}=1$ when performing robust trajectory optimization in Section \ref{sec:result}. 
    Starting from $\x_{0}$, we perform 500 experiments where parameters $\btheta$, disturbances $\w_k$, and control trajectories within $\U$ are randomized. 
    For different $M$ and $n_{\textrm{adv}}$, we use $(\textbf{robUP!})$ and compare positional volume coverage, which is crucial to determine whether a given homotopy class of paths is feasible or not, due to obstacle avoidance constraints. 
    Results in Figure \ref{fig:spacecraft_tuning} show that increasing $(M,n_{\textrm{adv}})$ beyond $(100,1)$ does not lead to drastic improvements, compared to the volume of the experiment with largest $M$ and $n_{\textrm{adv}}$. It is thus a reasonable choice of hyperparameters for this application. 
    In Figure \ref{fig:adv:spacecraft:adv_effect_iters_all}, we provide a visualization of the effect of adversarial sampling for this system, where all samples are projected onto $x,y$.
    \vspace{-2mm}
\end{minipage}

\begin{figure}[!htb]
\begin{minipage}{0.6\linewidth}
\centering
\includegraphics[width=1\linewidth]{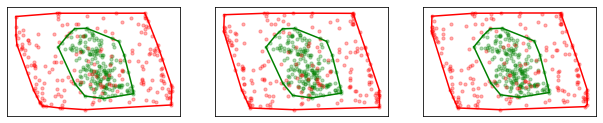}
\caption{Effect of adversarial sampling for $M\,{=}\,200$ and $n_{\textrm{adv}}\in\{1,2,5\}$ (in red) on the projection onto $x,y$ positions of the sampled reachable sets at time $k=12$ for the spacecraft system under uncertainty. In green, (\textbf{randUP}) is shown for reference.}
\label{fig:adv:spacecraft:adv_effect_iters_all}
\end{minipage}
\end{figure}

\begin{minipage}{\linewidth}
	\begin{wrapfigure}{R}{0.20\linewidth}
	\begin{minipage}{1\linewidth}
	\vspace{-5mm}
    \centering
	\includegraphics[width=0.8\linewidth]{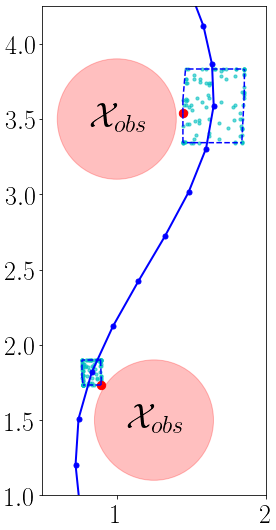}
	\caption{(\textbf{robUP!}) can be used for sensitivity analysis.}\label{fig:sensitivity}
	\vspace{-3mm}
	\end{minipage}%
	\end{wrapfigure}
	
    \vspace{-2mm}
    \paragraph{Adversarial sampling for sensitivity analysis} 
    Since (\textbf{robUP!}) actively searches for parameters and disturbances which lie outside the convex hull of the reachable states, it can be used to efficiently find parameters for which the system violates a property: a problem also known as falsification. 
    In particular, further insight can be gained from the solution of the robust path of the robust spacecraft planning problem from Section \ref{sec:result}. 
    In Figure \ref{fig:sensitivity}, the samples and convex hulls at times $k=8,14$ are shown. For each one, the sampled state $\smash{\x_k^j}$ closest to obstacles is shown in red. Their respective sampled parameters $\btheta^j$ and disturbances $\smash{\w_k^j}$ are equal to $\smash{(m^j,J_i^j)}\,{=}\,(7.1,0.075)$  for $k\,{=}\,8$, and $(7.3,0.065)$ for $k\,{=}\,14$, with saturated $\smash{\w_{k}^j}$ in opposite directions for the two particles. 
    This indicates that all variables influence the size of the reachable sets. 
    Moreover, a larger inertia does not necessarily correlate with smaller reachable sets, and both large and smaller values have an impact. 
    For future applications, this methodology could be used to analyze the sensitivity of more challenging dynamical systems with respect to different parameters, 
    guide sampling-based reachability analysis to respect a finite set of critical constraints, 
    and design robust online adaptation rules for learning-based controllers.
\end{minipage}

\section{Additional Experimental Details}\label{sec:appendix:lipschitz}
\subsection{Uncertainty Propagation using Lipschitz Continuity}\label{sec:lipschitz_derivations}

In this section, we detail our implementation of the Lipschitz-based uncertainty propagation method which we compare with in Section \ref{sec:result}. 
For these experiments, consider the dynamical system
\begin{equation}
\x_{k+1} = \f(\x_k) 
= \h(\x_k) + \g(\x_k), \quad \x_k\in\R^n.
\end{equation}
Note that we drop the dependence on the control input $\ac_k$ for conciseness, and since our comparisons concern a sequence of known open-loop controls. 
For simplicity, we assume $\h$ is an affine map, and $\g$ is Lipschitz continuous, such that for all $\x,\bmu\in\R^n$,
\begin{equation}
|g_i(\x)-g_i(\bmu)| \leq L_{g_i} \|\x-\bmu \|_2, \quad i=1,\mydots,n.
\end{equation}

The method presented in \citep{koller2018} consists of propagating ellipsoidal sets:
\begin{mydef}[Ellipsoidal Set]\label{def:ellipsoidal_conf_region}
	A set $\mathcal{B}(\bmu, \bQ)$, $\bmu\in\R^n,\bQ\in\mathbb{R}^{n\times n}, \bQ\succ 0$, is an ellipsoidal set if 
	\begin{equation}
	    \mathcal{B}(\bmu , \bQ) := \left\{ 
		\x \mid (\x-\bmu)^T \bQ^{-1} (\x-\bmu) \leq 1
		\right\}.
	\end{equation}
\end{mydef}

Assume that $\x_k\in\B(\bmu_k,\bQ_k)$. 
The problem consists of computing $\bmu_{k+1}, \bQ_{k+1}$ such that $\x_{k+1}\in\B(\bmu_{k+1},\bQ_{k+1})$. Generally, the reachable set of $\f(\x_k)$ given that $\x_k$ lies in an ellipsoidal set will not be an ellipsoidal set. However, an outer-approximation is sufficient for control applications where constraints satisfaction needs to be guaranteed. First, we compute the center of the ellipsoid as 
\begin{equation}
    \bmu_{k+1} = \f(\bmu_k).
\end{equation}
Since $\h$ is affine, its Jacobian does not depend on $\x$. Thus, we decompose the error to the mean as 
\begin{align*}
\x_{k+1}-\bmu_{k+1} &= \h(\x_k)-\h(\bmu_k)+\g(\x_k)-\g(\bmu_k)
\\
&= 
\nabla \h \,{\cdot}\, (\x_k-\bmu_k)+\g(\x_k)-\g(\bmu_k).
\end{align*}

First, given $\x_k\in\B(\bmu_k,\bQ_k)$, we have $\nabla \h \,{\cdot}\, (\x_k-\bmu_k)\in\B(\mathbf{0},\bQ_{\textrm{nom},k})$, with $\bQ_{\textrm{nom},k}=\h\bQ_k\h^T$. 

Second, we use the Lipschitz property of $\g$ to bound the approximation error component-wise as
\begin{align}
|g_i(\x_k)-g_i(\bmu_k)|\leq L_{g_i} \|\x_k-\bmu_k\|_2\leq L_{g_i} \lambda_{\max}(\bQ_k),
\end{align}
where $\lambda_{\max}(\bQ_k)$ denotes the largest eigenvalue of $\bQ_k$, and since $\x_k\in\B(\bmu_k,\bQ_k)$. This defines a rectangular set in which $\g(\x_k)-\g(\bmu_k)$ is guaranteed to lie, which can be outer-approximated by an ellipsoid as
\begin{equation}
\g(\x_k)-\g(\bmu_k)\in\B(\mathbf{0},\bQ_{\g_k}), \quad \text{where} \ \ \bQ_{\g_k}=n\cdot\textrm{diag}\big((L_{g_i}\lambda_{\max}(\bQ_k)^2), \ i=1,\mydots,n \big),
\label{eq:Qg_diag}
\end{equation}
where $\textrm{diag}(\mydots)$ denotes the diagonal matrix with diagonal components $(\mydots)$. 

Finally, the two terms can be combined as 
\begin{equation}
\x_{k+1}-\bmu_{k+1} \in
\B(\mathbf{0},\bQ_{\textrm{nom},k})
\oplus
\B(\mathbf{0},\bQ_{\g_k})
\subset
\B(\mathbf{0},\bQ_{k+1}),
\end{equation}
where $\bQ_{k+1}=\frac{c+1}{c}\bQ_{\textrm{nom},k} + (1+c)\bQ_{\g_k}$, with $c=\sqrt{\textrm{Tr}(\bQ_{\textrm{nom},k}/\textrm{Tr}(\bQ_{\g_k})}$, and $\textrm{Tr}(\cdot)$ denotes the trace operator. Finally, combining the terms above and by linearity,
\begin{equation}
\x_{k+1} \in \B(\bmu_{k+1},\bQ_{k+1}).
\end{equation}

Starting from $\x_0\in\B(\bmu_0,\bQ_0)$, and applying this recursion for all $k=0,\mydots, N{-}1$, this method enables the computation of sequence of sets which outer approximate the true reachable sets of the nonlinear system, given known upper-bounds for the Lipschitz constant of the dynamics.

\newpage

\subsection{Neural network experiment and comparisons}

\paragraph{Neural network training} 
To simplify training, aid generalization, and simplify comparisons with the Lipschitz-based method (see \ref{sec:appendix:lipschitz}), we decompose the dynamics as $\x_{k+1}=\f(\x_k,\ac_k)=\h(\x_k)+\g(\x_k,\ac_k)$, where $\h(\x_k)=\x_k$ is a known nominal term which captures prior knowledge about the system, and $\g$ is unknown and needs to be learned by the neural network. 
We opt for a feed-forward network architecture with $2$ hidden layers of width $128$ each, with $\tanh(\cdot)$ activation functions. 
To train the neural network, we re-sample states and controls at each training step as described in the next section, giving rise to the tuples $\{(\x_k,\ac_k,\x_{k+1})^b\}_{b=1}^B$. We use a single-step quadratic loss $\smash{\sum_{b=1}^B \|\x_{k+1}^b-\g(\x_k^b,\ac_k^b)\|_2^2}$, a batch size $B=20$, and include standard $L_2$-regularization with weight $10^{-6}$. All code is written using PyTorch \citep{paszke2017automatic}, and the model is trained using \textit{Adam} \citep{AdamPyTorch}, with an initial learning rate of $0.02$, and a decay factor of $(1-10^{-6})$ every gradient descent step. 
After 10k training steps, the model achieves a loss of around $10^{-7}$ on the validation dataset. We perform further validation through multi-steps rollouts of the system (over 20 timesteps). By adopting a model architecture which leads to minimal error, we are able to compare the reachable sets of the true linear system with those of the neural network, and evaluate volume coverage. 

\paragraph{Randomization} To provide comparisons of volume coverage, we perform $B\,{=}\,100$ experiments where we randomize initial states $\x_0^b\,{\in}\,\X_0\,{=}\,\B(\bmu_0,\bQ_0)$ and open-loop control trajectories $\ac^b\,{=}\,(\ac_0^b,\mydots,\ac_{N-1}^b)$, $b\,{=}\,1,\mydots,B$. 
We sample $\smash{\bmu_0^b\,{=}\,(\bpos_0^b,\bv_0^b)}$ with $\bpos_{0,i}^b\sim\textrm{Unif}(-5,5)$,  $\bv_{0,i}^b\sim\textrm{Unif}(-1,1)$, and constant $\smash{\bQ_0^b=10^{-3}\textrm{diag}([10,10,2,2])}$. 
We sample controls as $\smash{\ac_k^b=\bar{\ac}^b+\delta\ac_k^b}$ 
with $\bar{\ac}^b_i\sim\textrm{Unif}(-0.4,0.4)$ and $\delta\ac_{k,i}^b\sim\textrm{Unif}(-0.02,0.02)$ for neural network training, and with 
$\bar{\ac}^b_i\sim\textrm{Unif}(-0.1,0.1)$ and $\delta\ac_{k,i}^b\sim\textrm{Unif}(-0.005,0.005)$ for evaluation in Figure \ref{fig:results:nn_lip}.

\paragraph{Computation time}
In average, for this neural network, (
(\textbf{randUP})${}^{\scalebox{.7}{$1\textrm{k}$}}$, 
(\textbf{$\cdot$})${}^{\scalebox{.7}{$3\textrm{k}$}}$,
(\textbf{$\cdot$})${}^{\scalebox{.7}{$10\textrm{k}$}}$,
(\textbf{robUP!})${}^{\scalebox{.7}{$1\textrm{k}$}}_{\scalebox{.7}{$1$}}$
) require 
(9,28,120,648) ms on a laptop with an i7-6700 CPU (2.60GHz) and 8GB of RAM. 
We did not optimize our implementation. Performing operations in parallel and using a GPU would further accelerate both methods.

\paragraph{Lipschitz method}
In our experiments, we use $\h(\x)=\x$, and train a neural network to represent $\g(\x,\ac)$. To reduce conservatism in our comparisons, we use the Lipschitz constant of the true linear system  $\x_k = (\pos_k,\vel_k) \in \R^4$, $\ac_k \in \R^2$, $\pos_{k+1} = \pos_k + \vel_k$, and $\vel_{k+1} = \vel_k + \ac_k$, which are given as $L_{g_i}=1$ for $i=1,2$, and $L_{g_i}=0$ for $i=3,4$.

\paragraph{Volume computation}
Finally, the volume of the ellipsoidal sets $\B(\bmu,\bQ)$ are computed in closed form as
\vspace{-2mm}
\begin{equation}
\textrm{Vol}(\B(\bmu,\bQ)) = 
\frac{\pi^{n/2}}{\Gamma(n/2+2)}\frac{1}{\sqrt{\textrm{det}(\bQ^{-1})}},
\end{equation}
where $\Gamma(\cdot)$ is the standard gamma function of calculus  \citep{Sun2004}.

\section{Robust Trajectory Optimization with Sampling-based Convex Hulls}
\paragraph{Problem formulation} 
This section proposes a method to perform (approximately) robust trajectory optimization using sampling-based reachability analysis, and sequential convex programming (SCP). Specifically, we extend the method presented in \citep{LewBonalli2020} to leverage sampling-based convex hulls. 

In the following, we use the same assumptions and notations outlined in Section \ref{sec:problem}. 
The goal consists of computing an open-loop trajectory $(\x_{0:N},\ac_{0:N{-}1})$ which satisfies all constraints $(\x_k\in\Xsafe,\ac_k\in\U)$, for any bounded uncertain parameter $\btheta\in\Theta$ and disturbances $\w_k\in\W$. The initial state $\x_0$ is uncertain and lies within a known initial bounded set $\X_0$, and the final state $\x_N$ should lie within the  final goal region $\Xgoal$. 
The trajectory should minimize the step and final cost functions  $l:\X\times\U{\rightarrow}\R, \  l_f:\X{\rightarrow}\R$ (e.g., fuel consumption, and final velocity). 
To make this problem tractable, given a sequence of control inputs $\ac\,{=}\,(\ac_0,\mydots,\ac_{N{-}1})$, we define the nominal trajectory $\bmu\,{=}\,(\bmu_0,\mydots,\bmu_N)$, from a fixed $\bmu_0\,{\in}\,\X_0$, 
as 
\begin{equation}
\bmu_{k+1}=\f(\bmu_k,\ac_k,\bar{\btheta},\bar{\w}_k), \quad
\bar{\btheta}\in\Theta, \ \  \bar{\w}_k\in\W, 
\label{eq:nominal_traj}
\end{equation}
where $\bar{\btheta}$ and $\bar{\w}_k$ are fixed nominal parameters and disturbances. 
Given this nominal trajectory, we aim to minimize the cost of the nominal trajectory $\smash{\sum_{k=0}^{N-1} \cost(\bmu_k, \ac_k)+l_f(\bmu_N)}$, subject to all constraints defined above. We define the following robust optimal control problem: 
\\

\centerline{\textbf{Robust Optimal Control Problem}
\vspace{-6mm}}
\begin{subequations}
\label{eq:full_problem}
\begin{align}
\hspace{-10mm}\mathop{\text{min}}_{\ac_{0:N-1}}
\qquad &
\sum_{k=0}^{N-1} \cost(\bmu_k, \ac_k) + l_f(\bmu_N)
\label{eq:cost:measure}
\\
\hspace{-10mm}\text{subject to}\qquad & \x_{k+1} = \f(\x_k,\ac_k,\btheta,\w_k),  
\quad  \w_k\in\W, \ \ \btheta\in\Theta,
\quad  k\,{=}\,0, \mydots,N\,{-}\,1,
\\
& 
\bigwedge_{k=1}^N\big(\x_{k} \in \Xsafe\big) \ \cap \ 
\bigwedge_{k=0}^{N-1}\big(\ac_{k} \in \U\big) \  \cap \ 
\big(\x_{N} \in \Xgoal\big) \  \cap \ \big(\x_{0} \in \X_0\big)
.
\label{eq:robust_constraints_orig}
\end{align}
\end{subequations} 

%
Using the reachable sets $\smash{\{\X_k\}_{k=0}^N}$ defined in \eqref{eq:reach_set_onestep}, and the nominal trajectory in \eqref{eq:nominal_traj}, it is possible to show that the previous problem is equivalent to the following problem:
\\

\centerline{\textbf{Reachability-Aware Optimal Control Problem}
\vspace{-6mm}
}
\begin{align}
\mathop{\text{min}}_{\ac_{0:N-1}}
\ \  
\sum_{k=0}^{N-1} \cost(\bmu_k, \ac_k){+}l_f(\bmu_N) 
\quad
\text{s.t.}
\   \
\bigwedge_{k=1}^N \X_{k} \,{\subset}\, \Xsafe, 
\ \ \
\bigwedge_{k=0}^{N-1}\ac_{k} \,{\in}\, \U,
\ \  \
\X_{N} \,{\subset}\, \Xgoal,
\label{eq:reachabilityaware_ocp}
\end{align}
where $\smash{\{\X_k\}_{k=0}^N}$ depend on the chosen sequence of controls $\ac_{0:N-1}$, and are computed from $\X_{0}$. In this work, the reachable sets are approximated using either (\textbf{randUP}) or (\textbf{robUP!}), which yield convex hulls. However, we show next that computing the convex hull of the samples $\smash{\x_k^j}$ is not always necessary to reformulate common constraints found in robotic applications. 

\paragraph{Sequential convex programming (SCP) and constraints reformulation}
In this work, we leverage SCP to solve \eqref{eq:reachabilityaware_ocp}. The SCP technique consists of iteratively formulating convex approximations of \eqref{eq:reachabilityaware_ocp}, and solving each sub-problem using convex optimization. Specifically, at each iteration $(i{+}1)$, the previous nominal solution $(\bmu^i,\ac^i)$ is used to linearize all constraints, and reformulate \eqref{eq:reachabilityaware_ocp} as a quadratic program with linear constraints which can be efficiently solved using \textbf{OSQP} \citep{StellatoBanjacEtAl2017}. The solution of this problem, denoted as $(\bmu^{i{+}1},\ac^{i{+}1})$, should approach the solution of \eqref{eq:reachabilityaware_ocp} at convergence, which can be assessed by evaluating $\|\bmu^{i{+}1}-\bmu^i\|+\|\ac^{i{+}1}-\ac^i\|$. In this work, we use the open-source SCP procedure presented in \citep{LewBonalli2020}\footnote{The implementation of \citep{LewBonalli2020} is available at \url{github.com/StanfordASL/ccscp}.}, which includes trust region constraints and additional parameter adaptation rules to encourage convergence to a locally optimal solution to the original non-convex problem. 

Next, we present simple methods to use sampling-based reachable sets to reformulate constraints. First, consider dimension-wise constraints of the form
\begin{equation}
\x_{\min,i}\leq\x_{k,i}(\ac)\leq\x_{\max,i}, \ \ \ \forall (\x_0,\btheta,\w_{0:k{-}1})\,{\in}\, \X_0\times\Theta\times\W^k,
\label{eq:constraint_max_dimensionwise}
\end{equation}
where $i\in\{1,\mydots,n\}$. For instance, such constraints can represent velocity bounds, and are incorporated within $\x_k\in\Xsafe$ in  \eqref{eq:robust_constraints_orig}. 
To reformulate such constraints using the reachable set $\X_k(\ac)$ and the nominal trajectory $\bmu$, it suffices to compute its outer-bounding rectangle $\Delta_k(\ac)$, defined as $\Delta_k=\{\x_k \,| \, |\x_{k,i}-\bmu_{k,i}|\leq \delta_{k,i}\}$. 
Using (\textbf{randUP}) or (\textbf{robUP!}), the bounds  $\smash{\delta_{k,i}=\max_{j}|\x_{k,i}^j-\bmu_{k,i}|}$ can be efficiently computed. 
With these rectangular sets, \eqref{eq:constraint_max_dimensionwise} is conservatively reformulated as 
\begin{equation}
\x_{\min,i}+\delta_{k,i}(\ac)\leq\bmu_{k,i}(\ac)\leq\x_{\max,i}-\delta_{k,i}(\ac).
\end{equation}

Next, linear constraints $\ba\,{\cdot}\,\x_{k,1:s}\,{\leq}\, b$ on the $s$ first components of $\x_k$, with $\ba\in\R^s,b\in\R$ can be reformulated using outer-bounding ellipsoidal sets for the $s$ first components of $\x$: 
$\mathcal{B}^s(\bmu_k , \bQ_k) = \left\{ 
\x \mid (\x-\bmu_{k,1:s})^T \bQ_k^{-1} (\x-\bmu_{k,1:s}) \leq 1
\right\}$. 
Following derivations in Section \ref{sec:lipschitz_derivations}, $\bQ_k$ can be computed as $\smash{\bQ_k\,{=}\,s\,{\cdot}\,\textrm{diag}(\delta_i^2, \, i\,{=}\,1,\mydots,s)}$, yielding an ellipsoidal set $\B^s$ which outer-approximates the reachable set $\X_k$ (projected onto its  $s$ first dimensions). Using this ellipsoidal set, the linear constraint $\ba\,{\cdot}\,\x_{k,1:s}\,{\leq}\, b$ is conservatively reformulated as $\smash{\ba^T\bmu_{k,1:s}+(\ba^T\bQ_k\ba)^{1/2}\,{\leq}\, b}$, following similar derivations as in \citep{LewBonalli2020,koller2018}. Finally, this constraint is linearized, to yield a convex reformulation of the original constraints  \eqref{eq:robust_constraints_orig}, and obtain a convex quadratic program with linear constraints which can be solved using \textbf{OSQP}. 
For the spacecraft planning problem in Section \ref{sec:result}, such linear constraints are obtained from non-convex obstacle avoidance constraints, by expressing them using the signed distance function, and linearizing the expression. This procedure is described in \citep{LewBonalli2020}, where it is shown that it is conservative for general convex obstacles.

\end{document}